\documentclass[aps,pra,reprint,nofootinbib,groupedaddress,twoside,showpacs]{revtex4-1}

\usepackage{amsmath,amssymb}
\usepackage{amsthm}
\usepackage{hyperref} % create hyperlinks
\usepackage{acronym}  % make an acronym
\usepackage{color}  % make colors
\usepackage{amsfonts}
\usepackage{mathrsfs}
\usepackage{graphicx}
\usepackage[on]{psfrag}
\usepackage{array}
\usepackage{cases}
\usepackage{tabularx}
\usepackage{url}
\usepackage{soul}
\usepackage{theoremref}
\usepackage{thmtools}
\usepackage{lettrine}
\usepackage{winsnotation}

\makeatletter
\renewcommand*\env@matrix[1][\arraystretch]{%
  \edef\arraystretch{#1}%
  \hskip -\arraycolsep
  \let\@ifnextchar\new@ifnextchar
  \array{*\c@MaxMatrixCols c}}
\makeatother

\newcommand{\mybmatrix}[4] {\setlength{\arraycolsep}{#2}\left[\hspace{#1}\begin{matrix}[#3]#4\end{matrix}\hspace{#1}\right]\setlength{\arraycolsep}{1mm}}

\newcommand*{\QEDB}{\hfill\ensuremath{\square}}%

\declaretheoremstyle[
spaceabove=8pt, spacebelow=8pt,
headfont=\normalfont\bfseries,
notefont=\normalfont, notebraces={(}{)},
headpunct={:},
bodyfont=\normalfont,
postheadspace=0.5em,
]{mynote}
\declaretheorem[style=mynote]{Remark}
\declaretheorem[style=mynote,name=Theorem]{Thm}
\declaretheorem[style=mynote,name=Lemma]{Lem}
\declaretheorem[style=mynote,name=Proposition]{Prop}
\declaretheorem[style=mynote,name=Algorithm]{Algno}

\newcommand{\paperTitle}{Adaptive recurrence quantum entanglement distillation \\for two-Kraus-operator channels}

\begin{document}
%\markboth{\small PLEASE DO NOT DISTRIBUTE WITHOUT THE WRITTEN CONSENT OF THE AUTHORS \version}{\small Ruan, Dai, and Win: \paperTitleMarkboth}
%\markboth{Please do not distribute the manuscript without the written consent of the authors. \version}{Ruan, Dai, and Win: \paperTitleMarkboth}

%\thispagestyle{empty}
{\color{white} 
\fontsize{0pt}{0pt}\selectfont
\begin{acronym}
\acro{SVD}{singular value decomposition}\vspace{-16mm}
\acro{LOCC}{local operations and classical communication}\vspace{-16mm}
\acro{QED}{quantum entanglement distillation}\vspace{-16mm}
\acro{QEC}{quantum error correction}\vspace{-16mm}
\acro{w.r.t.}{with respect to}\vspace{-16mm}
\acro{FP}{fidelity-prioritized}\vspace{-16mm}
\acro{PP}{probability-prioritized}\vspace{-16mm}
\acro{QPA}{quantum privacy amplification}\vspace{-16mm}
\acro{TKO}{two-Kraus-operator}\vspace{-16mm}
\acro{RSSP}{remote shared-state preparation}\vspace{-16mm}
\end{acronym}}
%\clearpage

% Use the \preprint command to place your local institutional report
% number in the upper righthand corner of the title page in preprint mode.
% Multiple \preprint commands are allowed.
% Use the 'preprintnumbers' class option to override journal defaults
% to display numbers if necessary
%\preprint{}

%Title of paper
%\newpage
\setcounter{page}{1}
%\title{Adaptive Recurrence Quantum Distillation for Two-Kraus-Operator Channels}
%\title{\paperTitle}
\title{\paperTitle}
% repeat the \author .. \affiliation  etc. as needed
% \email, \thanks, \homepage, \altaffiliation all apply to the current
% author. Explanatory text should go in the []'s, actual e-mail
% address or url should go in the {}'s for \email and \homepage.
% Please use the appropriate macro foreach each type of information

% \affiliation command applies to all authors since the last
% \affiliation command. The \affiliation command should follow the
% other information
% \affiliation can be followed by \email, \homepage, \thanks as well.
\author{Liangzhong~Ruan, Wenhan~Dai, and Moe~Z.~Win}
%\email[]{Your e-mail address}
%\homepage[]{Your web page}
%\thanks{}
%\altaffiliation{}
\affiliation{Laboratory for Information and Decision Systems, Massachusetts Institute of Technology\\
77 Massachusetts Avenue, Room 32-D608, Cambridge, MA 02139}

%Collaboration name if desired (requires use of superscriptaddress
%option in \documentclass). \noaffiliation is required (may also be
%used with the \author command).
%\collaboration can be followed by \email, \homepage, \thanks as well.
%\collaboration{}
%\noaffiliation

\date{\today}

\begin{abstract}
% insert abstract here
Quantum entanglement serves as a valuable resource for many important quantum operations.
A pair of entangled qubits can be shared between two agents by first preparing a maximally entangled qubit pair at one agent, and then sending one of the qubits to the other agent through a quantum channel. 
In this process, the deterioration of entanglement is inevitable since the noise inherent in the channel contaminates the qubit.
To address this challenge, various \ac{QED} algorithms have been developed.
% to generate qubit pairs with high fidelity with respect to the targeted entangled state from many shared contaminated pairs.
Among them, recurrence algorithms have advantages in terms of implementability  and robustness.
However, the efficiency of recurrence \ac{QED} algorithms has not been investigated thoroughly in the literature.
This paper put forth two recurrence  \ac{QED} algorithms that adapt to the quantum channel to tackle the efficiency issue.
The proposed algorithms have guaranteed convergence for quantum channels with two Kraus operators, which include phase-damping and amplitude-damping channels.
Analytical results show that the convergence speed of these algorithms is improved from linear to quadratic and one of the algorithms achieves the optimal speed.
%These algorithms improve the convergence speed of fidelity achieved in from linear to quadratic, and one of them achieves the optimal fidelity.
Numerical results confirm that the proposed algorithms significantly improve the efficiency of \ac{QED}.
\end{abstract}

% insert suggested PACS numbers in braces on next line
% insert suggested keywords - APS authors don't need to do this

%\maketitle must follow title, authors, abstract, \pacs, and \keywords
\maketitle

% body of paper here - Use proper section commands
% References should be done using the \cite, \ref, and \label commands
\acresetall             % reset the acronyms

%---------------------------------------------------------------------------%
%                                Introduction                               %
%---------------------------------------------------------------------------%
\section{Introduction}
% Put \label in argument of \section for cross-referencing
%\section{\label{}}

Quantum entanglement shared by remote agents serves as a valuable resource for many important quantum communication operations \cite{Llo:03,HorHorHorHor:09,UlaFedPusKurRalLvo:J15}, such as secret key distribution \cite{Eke:91,KoaPre:03,GotLoLutPre:04}, dense coding \cite{BenSte:92,WanDenLiLiuLon:05,BarWeiKwi:08}, and teleportation \cite{BenBraCreJozPerWoo:93,NieKniLaf:98,GotChu:99}.
With the assistance of entanglement, the capacity of quantum channels can be increased, particularly when the channels are very noisy \cite{BenShoSmoTha:99,BenShoSmoTha:02,Sho:b04,Hol:12,HolShi:15}.
Entanglement also enables quantum relay, and therefore is a keystone of long-distance quantum communication \cite{DurBriCirZol:99,SanSimRieGis:11}.
To establish entanglement between two remote agents, one agent can locally generate a maximally entangled qubit pair and send one of the qubits to the other agent through a quantum channel.
However, the noise inherent in the channel will contaminate the qubit during the transmission, thereby deteriorating the entanglement.
To address this problem, \ac{QED} algorithms \cite{BenBraPopSchSmoWoo:96,DeuEkeJozMacPopSan:96,BriDurCirZol:98,OpaKur:99,MunOrs:09,DehVanDeMVer:03,VolVer:05,HosDehDeM:06,BenDivSmoWoo:96,Mat:03,AmbGot:06,WatMatUye:06} have been proposed to generate qubit pairs in the targeted maximally entangled state from many shared contaminated ones using \ac{LOCC}.

In the pioneering work \cite{BenBraPopSchSmoWoo:96}, two influential \ac{QED} algorithms were proposed and are now known as the recurrence algorithm and the asymptotic algorithm.
Recurrence algorithms \cite{DeuEkeJozMacPopSan:96,BriDurCirZol:98,OpaKur:99,MunOrs:09} operate separately on every two qubit pairs, improving the quality of entanglement in one pair at the expense of the other pair, which is then discarded.
The algorithms keep repeating this operation to progressively improve the quality of entanglement in the kept qubit pairs.
%The remaining pairs repeatedly undergoe the same process, progressively reaching a higher quality of entanglement.
Asymptotic algorithms \cite{DehVanDeMVer:03,VolVer:05,HosDehDeM:06} operate on a large number of qubit pairs, detecting ones that are not in the targeted state by measuring a subset of the qubit pairs, and then transforming those in the undesired state to the targeted state.
Later, it was recognized that there is a duality between \ac{QED} and \ac{QEC} \cite{BenDivSmoWoo:96} when one-way classical communication is involved.
The connection between \ac{QED} and \ac{QEC}  was further explored in scenarios involving two-way classical communication, enabling code-based \ac{QED} algorithms \cite{Mat:03,AmbGot:06,WatMatUye:06}.
These algorithms operate on a few qubit pairs, look for the error syndrome using measurements specified by the error correction code, and then correct the errors to restore entanglement.

Asymptotic algorithms have revealed important theoretical insights, but these algorithms require agents that have the capability of processing a large number of qubits.
Code-based algorithms require agents to have the capability of processing only a few qubits, but the number of errors that can be corrected is limited by the Hamming distance of the codewords.
Designing \ac{QEC} codes with large Hamming distance is challenging since the creation of information redundancy, the main mechanism adopted in classical error correction codes, is not possible in quantum codes due to the no-cloning theorem \cite{Sho:95,Got:96,Kni:05,Got:09}.
%For instance, even when agents have the capability to process arbitrarily large number of qubits, no code when the error rate exceeds $37.9\%$
%Studies have shown that accurate quantum computation is possible with error rate as high as 3\% \cite{}. While this is an important result for quantum computation, the error rate in quantum communication can be significantly higher than 3\%.
Hence, these algorithms do not apply to scenarios with strong noise in the channel.
Recurrence algorithms require agents to have the capability of processing only a few qubits and can generate maximally entangled qubit pairs even in strong noise scenarios.
This is because the recurrence algorithms can mitigate stronger noise by performing more rounds of distillations.
In fact, the recurrence algorithm proposed in  \cite{BenBraPopSchSmoWoo:96} can distill contaminated qubit pairs into
maximally entangled qubit pairs as long as the initial fidelity of the contaminated qubit pairs \ac{w.r.t.} the targeted state is greater than $0.5$.\footnote{In \cite{HorHor:96,HorHorHor:97}, the authors proved that a state of qubit pairs is distillable if and only if its fidelity \ac{w.r.t.}  \emph{some} maximally entangled state is greater than $0.5$. 
Assuming that the agents have knowledge of the channel to properly select the targeted maximally entangled state, 
then the sufficient distillability condition provided in  \cite{BenBraPopSchSmoWoo:96} is equivalent to that proposed in \cite{HorHor:96,HorHorHor:97}.}

Building large-scale quantum circuits operating on many qubits is challenging \cite{Unr:95,Div:00,Suo:12}; 
even the error rates of two-qubit gates are significantly higher than those of one-qubit gates \cite{LadJelLafNakMonBri:10}.
In this perspective, recurrence algorithms are favorable for implementation as they require operations only on a few (typically one or two) qubits and are robust to strong noise in the quantum channel.
On the other hand, since at least half of the entangled qubit pairs are discarded in each round of distillation, the efficiency of the recurrence algorithms decreases dramatically with the number of rounds.\footnote{The efficiency of \ac{QED} algorithms is measured in terms of \emph{yield}, which is defined as the ratio between the number of maximally entangled output qubit pairs and the number of contaminated input qubit pairs.}
To address this challenge, the \ac{QPA} algorithm was proposed in \cite{DeuEkeJozMacPopSan:96}, and
was shown numerically to require fewer rounds of distillation than the algorithm in \cite{BenBraPopSchSmoWoo:96} for contaminated qubit pairs with a specific set of initial states.
%numerical results showed that \ac{QPA} algorithm required fewer rounds of distillation than the algorithm in \cite{BenBraPopSchSmoWoo:96} for some specific initial state of the contaminated qubit pairs.
However, the performance of \ac{QPA} algorithm was not characterized analytically.
In fact, another set of initial states was found in \cite{OpaKur:99} for which the \ac{QPA} algorithm was less efficient than the algorithm in \cite{BenBraPopSchSmoWoo:96}.
In \cite{OpaKur:99}, the design of distillation operations was formulated into an optimization problem,
which was inherently non-convex, and consequently, the optimal solution was not found.

We envision that a key enabler to designing efficient recurrence \ac{QED} algorithms is to make them adaptive to quantum channel.
Intuitively, compared to general algorithms,
\ac{QED} algorithms that adapt to channel-specific noise will better mitigate such noise and hence distill more efficiently.
In fact, it has been observed that knowing the channel benefits the performance of quantum error recovery \cite{FleShoWin:J07}, and channel-adaptive \ac{QEC} schemes that outperform prior ones \cite{FleShoWin:J08,FleShoWin:J08a} have been designed.

In this paper, we focus on \ac{TKO} channels, a class that covers several typical quantum channels, e.g., phase-damping and amplitude-damping channels.
The phase-damping channels describe the decoherence process of a photon traveling through a waveguide,
and the amplitude-damping channels model the decay of an excited atom due to spontaneous emission \cite[Sec. 3.4]{Pre:B}, \cite[Sec.8.3]{NieChu:B00}.
To achieve efficient distillation, we develop two adaptive recurrence \ac{QED} algorithms which adapt to the channel by employing a \ac{RSSP} method.\footnote{This method is akin to remote state preparation methods \cite{BenDivShoSmoTerWoo:01,BenHayLeuShoWin:05}, in which two remote agents employ \ac{LOCC} to prepare a quantum state at one of the agents. The proposed \ac{RSSP} method uses \ac{LOCC} to prepare a state shared by both agents.}
 The contributions of this work are
\begin{itemize}
\item characterization of the structure of \ac{TKO} channels;
\item characterization of the optimal fidelity that can be achieved by performing \ac{LOCC} on two qubits pairs affected by \ac{TKO} channels;
\item design of adaptive recurrence \ac{QED} algorithms which improve the convergence speed of fidelity from linear to quadratic, and one of them achieves the optimal speed.
\end{itemize}

\noindent{\bf Notations:} $a$, $\V{a}$, and $\M{A}$ represent scalar, vector, and matrices, respectively.
$\mathrm{pha}\{\cdot\}$ denotes the phase of a complex number.
$(\cdot)^\dag$, $\mathrm{rank}\{\cdot\}$, $\det\{\cdot\}$  and $\mathrm{tr}\{\cdot\}$, denote the Hermitian transpose, rank, determinant, and trace of a matrix, respectively. $\mathrm{tr}_{i,j}\{\cdot\}$ denotes the partial trace \ac{w.r.t.} to the $i$-th and $j$-th qubits in the system. $\mathrm{span}(\cdot)$ denotes the linear space spanned by a set of vectors. $\mathbb{I}_n$ denotes $n\times n$ identity matrix, and $\imath$ is the unit imaginary number.

\section{System Model}
Consider two remote agents, Alice and Bob, connected by a quantum channel and a two-way classical channel. When Alice transmits a qubit with density matrix $\V{\rho}_0$, the density matrix of the qubit received by Bob is given by
\begin{align}
\V{\rho}=\sum_{k=1}^{K}\M{C}_k\V{\rho}_0\M{C}^\dag_k,\label{eqn:channel}
\end{align}
where the Kraus operators $\{\M{C}_k\}$ representing the noisy quantum channel  satisfy
\begin{align}
\sum_{k=1}^{K}\M{C}^\dag_k\M{C}_k=\mathbb{I}_{2}.\label{eqn:channel_identity}
\end{align}
Since a qubit is a 2-dimensional system, the number of the Kraus operators $K\le 2^2=4$.\footnote{Quantum operators on $n$-dimensional systems are $n\times n$ matrices, and hence lies in an $n^2$ dimensional space. Thus, from \cite[Sec 3.3]{Pre:B}, if a channel for such systems is represented with more than $n^2$ operators, there always exists an equivalent representation with no more than $n^2$ non-zero operators.} 
When $K=1$, the channel is noiseless, and hence \ac{QED} is not needed.
For the class of \ac{TKO} channels, $K=2$.

Suppose Alice and Bob wish to obtain maximally entangled qubits pairs with density matrix $\V{\rho}_0=|\Phi^+\rangle\langle\Phi^+|$, where $|\Phi^+\rangle = \frac{1}{\sqrt{2}}\left(|00\rangle + |11\rangle\right)$.
To achieve this task, Alice locally prepares qubits pairs, each with density matrix $\V{\rho}_0$,
then sends the second qubit in each pair through the noisy channel. Then the density matrix of the two remote qubits becomes
\begin{align}
\V{\rho}&=\sum^2_{k=1} (\mathbb{I}_{2}\otimes\M{C}_k) \, \V{\rho}_0\,  (\mathbb{I}_{2}\otimes\M{C}_k)^\dag.
\label{eqn:initialstate}
\end{align}
 Alice and Bob then adopt a recurrence \ac{QED} algorithm outlined in Fig.~\ref{fig_recurrence}.
In each round of distillation, the agents separately operate on every two qubit pairs kept in the previous round,
perform \ac{LOCC}, and attempt to improve the quality of entanglement in one of the qubit pairs  at the expense of the other pair, which is then discarded (agents may discard both pairs when this operation is unsuccessful). The objective of the algorithm is to generate qubit pairs with density matrix $\V{\rho}^*$ close to the targeted state, i.e.,
\begin{align*}
\langle \Phi^+| \,\V{\rho}\, |\Phi^+\rangle \approx 1\end{align*}
where $\langle \Phi^+| \,\V{\rho}\, |\Phi^+\rangle$ is the fidelity of a density matrix $\V{\rho}$ and the targeted state $|\Phi^+\rangle$.

\begin{figure}[t] \centering
\includegraphics[scale=0.45]{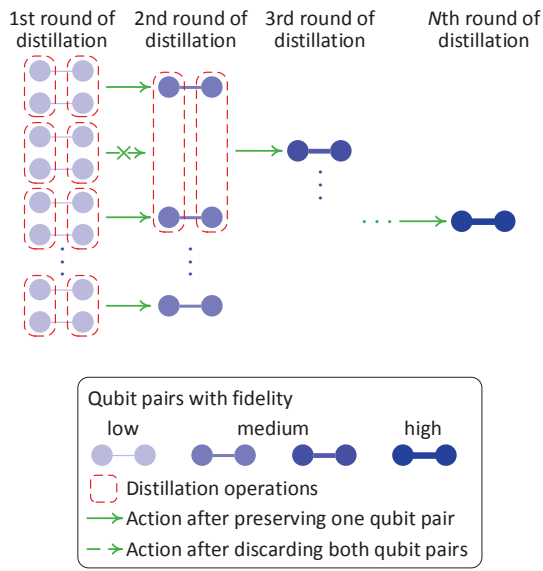}
\caption {The structure of recurrence \ac{QED} algorithms.}
\label{fig_recurrence}
\end{figure}

\section{Design of Adaptive Recurrence \ac{QED} Algorithm}
\subsection{Characterization of \ac{TKO} Channels}
Prior to designing adaptive \ac{QED} algorithms, it is crucial to understand the effect of noisy quantum channels on the entanglement between qubits.
This can be accomplished by determining the structure of the noisy quantum channels.
In particular, the structure of \ac{TKO} channels is provided by the following lemma.

\begin{Lem}[Structure of \ac{TKO} channels]\thlabel{lem:channel} For every single-qubit \ac{TKO} channel, there exist unitary matrices $\M{U},\V{V}\in\mathbb{C}^{2\times 2}$ and scalars $p\in[0,1]$, $\zeta\in[0,1]$, and $\eta\in\mathbb{C}$ with $|\eta|^2+\zeta^2=1$ such that the channel can be represented by
\begin{align}
{\M{C}}_1=\M{U}\begin{bmatrix}1&0\\0&\sqrt{1-p}\end{bmatrix}\V{V}^\dag,\qquad
{\M{C}}_2=\M{U}\begin{bmatrix}0&\eta\sqrt{p}\\0&\zeta\sqrt{p}\end{bmatrix}
\V{V}^\dag.
\label{eqn:channel_structure}
\end{align}
\end{Lem}
\begin{proof}
The proof is given in Appendix~\ref{pf_lem:channel}.
\end{proof}

Local unitary operations $\M{U}$ and $\V{V}$ do not affect the amount of entanglement \cite{VedPleRipKni:97}.
In particular, for any $\M{U}$ and $\V{V}$, Bob and Alice can respectively perform local unitary operations  $\M{U}^\dag$ and $\V{V}$ and obtain an equivalent channel with $\tilde{\M{U}}=\M{U}^\dag\M{U}=\mathbb{I}_2$ and
$\tilde{\V{V}}=\V{V}^\dag\V{V}=\mathbb{I}_2$.
Therefore, without loss of generality, $\M{U}$ and $\V{V}$ are assumed to be $\mathbb{I}_2$ in the following analysis. In this case, the channel can be represented by
\begin{align}
{\M{C}}_1=\begin{bmatrix}1&0\\0&\sqrt{1-p}\end{bmatrix},\qquad
{\M{C}}_2=\begin{bmatrix}0&\eta\sqrt{p}\\0&\zeta\sqrt{p}\end{bmatrix}.
\label{eqn:channel_structure2}
\end{align}

\begin{Remark}[The effect of \ac{TKO} channel on entanglement]\label{rem:chanel_para}
From \eqref{eqn:channel_structure2}, $\mathrm{rank}\{\M{C}_2\}=1$. Hence, $\M{C}_2$ can be written as in  \eqref{eqn:C2}.
In this case, for any qubit pairs with density matrix $\V{\sigma}_0$, conditional on that $\M{C}_2$ operates, 
the density matrix of the qubit pair after the second qubit goes through the channel becomes
\begin{align*}
\V{\sigma}&=\frac{(\mathbb{I}_{2}\otimes\M{C}_2) \, \V{\sigma}_0 \, (\mathbb{I}_{2}\otimes\M{C}_2)^\dag}
{\mathrm{tr}\{(\mathbb{I}_{2}\otimes\M{C}_2) \, \V{\sigma}_0 \, (\mathbb{I}_{2}\otimes\M{C}_2)^\dag\}}
=\V{\sigma}_{\mathrm{A}}\otimes\V{\sigma}_{\mathrm{B}}\end{align*}
where $\V{\sigma}_{\mathrm{B}}=|i\rangle\langle i|$, with $|i\rangle=\eta|0\rangle+\zeta|1\rangle$.
%In this case, 
%\begin{align*}
%\V{\rho}&=\frac{(\mathbb{I}_{2}\otimes\M{C}_2)  \V{\rho}_0  (\mathbb{I}_{2}\otimes\M{C}_2)^\dag}
%{\mathrm{tr}\{(\mathbb{I}_{2}\otimes\M{C}_2)  \V{\rho}_0  (\mathbb{I}_{2}\otimes\M{C}_2)^\dag\}}=|\psi\rangle\langle\psi|,
%\end{align*}
%where $|\psi\rangle = (\langle j|0\rangle |0\rangle + \langle j|1\rangle |1\rangle)\otimes |i\rangle$. 
Therefore, when $\M{C}_2$ operates, $\V{\sigma}$ is a separable state, which implies that all entanglement between the two qubits is destroyed. 
The behavior of $\M{C}_2$ is mainly characterized by parameters 
$p$ and $\eta$.
\begin{itemize}
\item{\em The role of $p$:} $p\in[0,1]$ is the strength of $\M{C}_2$, which is proportional to the probability that $\M{C}_2$ operates on a qubit. The parameter $p$ can be thought of as the severity of noise in the channel since it characterizes the extent that the channel deteriorates the entanglement.
The larger the $p$, the more the entanglement is destroyed. 
In particular, all entanglement is preserved when $p=0$, and the opposite when $p=1$.
\item{\em The role of $\eta$:} $\arcsin(|\eta|)\in[0,\frac{\pi}{2}]$ is the angle between the image (i.e., $\mathrm{span}(\eta|0\rangle+\zeta|1\rangle)$) and the coimage (i.e., $\mathrm{span}(|1\rangle)$) of $\M{C}_2$. This characterizes the angle at which $\M{C}_2$ rotates the state of a qubit, and hence indicates the type of the channel. In particular, the channel is phase-damping when $\arcsin(|\eta|)=0$, i.e., $\eta=0$, and amplitude-damping when $\arcsin(|\eta|)=\frac{\pi}{2}$, i.e., $|\eta|=1$.
\end{itemize}
The channel properties described above are summarized in Fig.~\ref{fig_channels}.~\QEDB
\end{Remark}

\begin{figure}[t] \centering
\includegraphics[scale=0.42]{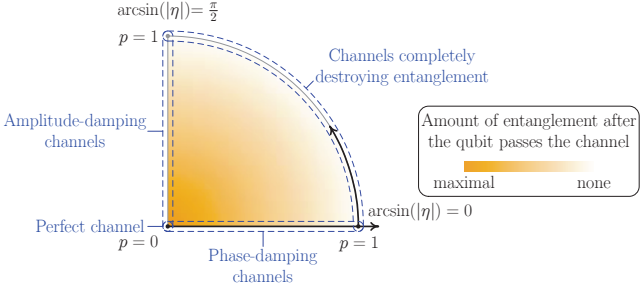}
\caption {Characterization of \ac{TKO} channels.}
\label{fig_channels}
\end{figure}

All entanglement is destroyed after the second qubit passes through a \ac{TKO} channel with $p=1$. 
Therefore, the interesting case for \ac{QED} is $p\in[0,1)$.
The following theorem characterizes the structure of the density matrix $\V{\rho}$ in this case.

\begin{Thm}[Structure of $\V{\rho}$]  \thlabel{thm:Str_rho} Consider the density matrix  of a qubit pair after the second qubit passes through a \ac{TKO} channel represented by \eqref{eqn:channel_structure2}. When $p<1$, there exist local unitary operators $\M{U}_{\mathrm A}$, $\M{U}_{\mathrm B}$ such that
\begin{align}
\check{\V{\rho}}=(\M{U}_{\mathrm A}\otimes\M{U}_{\mathrm B})\,\V{\rho}\,(\M{U}_{\mathrm A}\otimes\M{U}_{\mathrm B})^\dag = F|\mu\rangle\langle\mu| + (1-F)|\nu\rangle\langle\nu|
\label{eqn:int_dense_II}
\end{align}
where
\begin{align}
|\mu\rangle &= \alpha|00\rangle + \beta |11\rangle\label{eqn:mu}\\
|\nu\rangle &= \gamma|01\rangle + \delta e^{\imath \theta} |10\rangle\label{eqn:nu}
\end{align}
with $\theta$ a certain constant in $[0,2\pi)$, and
\begin{align}
F &= \frac{1}{2}+\frac{1}{2}\sqrt{(1-p)(1-|\eta|^2p)},\label{eqn:F_pzeta}\\
\alpha &= \sqrt{\frac{1}{2}+\frac{|\eta| p}{4F}},\qquad\hspace{10.5mm} \beta   = \sqrt{\frac{1}{2}-\frac{|\eta| p}{4F}},\label{eqn:alphabeta_pzeta}\\
\gamma  &= \sqrt{\frac{1}{2}-\frac{|\eta| p}{4(1-F)}},\qquad \delta = \sqrt{\frac{1}{2}+\frac{|\eta| p}{4(1-F)}}.\label{eqn:gammadelta_pzeta}
\end{align}
\end{Thm}
\begin{proof}
The proof is given in Appendix~\ref{pf_thm:Str_rho}.
\end{proof}

\begin{Remark}[The role of $\M{U}_{\mathrm{A}}$ and $\M{U}_\mathrm{B}$]
%Although local unitary operations  have no effect on the amount of entanglement \cite{VedPleRipKni:97}, 
Equations \eqref{eqn:int_dense_II}--\eqref{eqn:nu} show that by applying the properly designed local unitary operators $\M{U}_{\mathrm{A}}$ and $\M{U}_\mathrm{B}$, the density matrix $\V{\rho}$ can be transformed into $\check{\V{\rho}}$ with
all its eigenvectors written in the computational basis, i.e., $\{|0\rangle,|1\rangle\}$ via the Schmidt decomposition.
This transformation enables the simplification for both analysis and algorithm design in the following sections.
In particular, $\M{U}_{\mathrm{A}}$ and $\M{U}_\mathrm{B}$ will be employed by Alice and Bob respectively on their individual qubits before the recurrent distillation operations begin, and hence will be referred to as the pre-distillation unitary operators.~\QEDB
\end{Remark}

\subsection{Characterization of the Optimal Fidelity}

This section proves the optimal fidelity that can be achieved by performing appropriate \ac{LOCC}.
Consider recurrence \ac{QED} algorithms outlined in Fig.~\ref{fig_recurrence}, which 
\begin{itemize}
\item perform \ac{LOCC} on two qubit pairs with density matrix $\check{\V{\rho}}$ given by \eqref{eqn:int_dense_II}, and 
\item keep at most one pair.
\end{itemize}
Since operations performed by agents are local, they can be expressed as
$\V{N}^{(k)}_{\mathrm{A}}\otimes\V{N}^{(k)}_{\mathrm{B}}$, $k\in\{1,2,\ldots,K\}$, satisfying
$\sum_{k=1}^K\big(\V{N}^{(k)}_{\mathrm{A}}\big)^\dag\V{N}^{(k)}_{\mathrm{A}}=
\sum_{k=1}^K\big(\V{N}^{(k)}_{\mathrm{B}}\big)^\dag\V{N}^{(k)}_{\mathrm{B}}=\mathbb{I}_{4
}$. Without loss of generality, assume that agents keep the first qubit pair conditioned on the event that one of the first $\tilde{K}$ operators acts on the four qubits. Then after the \ac{LOCC}, the density matrix of the first qubit pair is given by
\begin{align*}
\breve{\V{\rho}}=\frac{\mathrm{tr}_{2,4}\big\{\sum_{k=1}^{\tilde{K}}(\V{N}^{(k)}_{\mathrm{A}}\otimes\V{N}^{(k)}_{\mathrm{B}})(\V{P}\check{\V{\rho}}\otimes\check{\V{\rho}}\,\V{P}^\dag)(\V{N}^{(k)}_{\mathrm{A}}\otimes\V{N}^{(k)}_{\mathrm{B}})^\dag\big\}}
{\mathrm{tr}\big\{\sum_{k=1}^{\tilde{K}}(\V{N}^{(k)}_{\mathrm{A}}\otimes\V{N}^{(k)}_{\mathrm{B}})(\V{P}\check{\V{\rho}}\otimes\check{\V{\rho}}\,\V{P}^\dag)(\V{N}^{(k)}_{\mathrm{A}}\otimes\V{N}^{(k)}_{\mathrm{B}})^\dag\big\}}
\end{align*}
where 
\begin{align*}
\V{P}=\mathbb{I}_2\otimes\big(|00\rangle\langle00|+|10\rangle\langle01|+|01\rangle\langle10|+|11\rangle\langle11|\big)\otimes\mathbb{I}_2
\end{align*} is the permutation operator that switches the second and third qubits. With this operator, the joint density matrix $\V{\rho}_{\mathrm{J}}=\V{P}\check{\V{\rho}}\otimes\check{\V{\rho}}\V{P}^\dag$ corresponds to four qubits, where the first two belong to Alice and last two belong to Bob.

Denote $F^*$ as the optimal fidelity that can be achieved with initial density matrix $\check{\V{\rho}}$ via all possible \ac{LOCC}, i.e.,
\begin{align}
F^*=\max_{\{\V{N}^{(k)}_{\mathrm{A}},\V{N}^{(k)}_{\mathrm{B}}\}^{\tilde{K}}_{k=1}\in\Set{F}}\langle \Phi^+| \,\breve{\V{\rho}}\, |\Phi^+\rangle\label{eqn:defoptF}
\end{align}
where $\Set{F}$ denotes the set of all possible \ac{LOCC}.

The characterization of the optimal fidelity $F^*$  is challenging because it involves general \ac{TKO} channels and arbitrary \ac{LOCC}.
These two issues are tackled by the following lemmas.
\thref{lem:simplifybound} characterizes the relationship between the $F^*$ for general \ac{TKO} channels and that for the special case of phase-damping channels.
%Note that all local operators are separable.
\thref{lem:operation} exploits the property of separable operators to determine
the set of attainable density matrices of the first qubit pair.% after arbitrary local operator has performed on two qubit pairs.

\begin{Lem}[Simplification to phase-damping]\thlabel{lem:simplifybound}
Express the optimal fidelity $F^*$ explicitly as a function of the density matrix parameters in  \eqref{eqn:F_pzeta}--\eqref{eqn:gammadelta_pzeta}, i.e., \begin{align*}
F^*=f(F, \alpha, \beta, \gamma, \delta, \theta).
\end{align*}
If the optimal fidelity for phase-damping channels is upper bounded by
\begin{align*}
f(F, \frac{1}{\sqrt{2}},  \frac{1}{\sqrt{2}}, \frac{1}{\sqrt{2}},  \frac{1}{\sqrt{2}}, 0)\le\frac{F^2}{F^2+(1-F)^2}, \quad\forall F\in(\frac{1}{2},1]
\end{align*}
then the optimal fidelity for generic \ac{TKO} channels satisfies 
\begin{align}
f(F, \alpha, \beta, \gamma, \delta, \theta)\le
\frac{F^2}{F^2+(1-F)^2(\frac{\gamma\delta}{\alpha\beta})^2}\label{eqn:F_lebound}
\end{align}
$\forall F$, $\alpha$, $\beta$, $\gamma$, $\delta$, and $\theta$ satisfying \eqref{eqn:F_pzeta}--\eqref{eqn:gammadelta_pzeta}.
\end{Lem}
\begin{proof}
The proof is given in Appendix~\ref{pf_lem:simplifybound}.
\end{proof}

\begin{Lem}[Density matrix after arbitrary seperable operation]\thlabel{lem:operation}
For phase-damping channels, after arbitrary seperable operator acts on two qubit pairs, the density matrix of the kept qubit pair
can be expressed as
\begin{align}
\breve{\V{\rho}}=\frac{\sum_{i=1}^{4}C_i\pmb{\psi}^{(i)}\pmb{\psi}^{(i)\dag}}
{\sum_{i=1}^{4}C_i\pmb{\psi}^{(i)\dag}\pmb{\psi}^{(i)}}\label{eqn:rho2}
\end{align}
where $C_1=F^2$, $C_2=C_3=F(1-F)$, $C_4=(1-F)^2$,
\begin{align}
\pmb{\psi}^{(i)}&=\begin{bmatrix}
w_{11} & x_{11} & y_{11} & z_{11}\\
w_{12} & x_{12} & y_{12} & z_{12}\\
w_{21} & x_{21} & y_{21} & z_{21}\\
w_{22} & x_{22} & y_{22} & z_{22}
\end{bmatrix}\V{v}^{(i)}\label{eqn:phi_3}
\end{align}
in which $s_{ij}$, $s\in\{w,x,y,z\}$, $i,j\in\{1,2\}$ are complex numbers, satisfying
\begin{align}
s_{11}s_{22}=s_{12}s_{21}\label{eqn:srelation}
\end{align}
$\V{v}^{(i)}$ is the $i$-th column of the unitary matrix
\begin{align}
\V{V}=\frac{1}{2}
\left[\begin{array}{rrrr}1&1&1&1\\1&-1&1&-1\\1&1&-1&-1\\1&-1&-1&1\end{array}\right]\label{eqn:U}
\end{align}
and $\sum_{i=1}^{4}C_i\pmb{\psi}^{(i)\dag}\pmb{\psi}^{(i)}>0$.
\end{Lem}
\begin{proof}
The proof is given in Appendix~\ref{pf_lem:operation}.
\end{proof}

With the issues of general \ac{TKO} channels and arbitrary \ac{LOCC} addressed, the optimal fidelity $F^*$ can now be characterized.
\begin{Thm}[Optimal Fidelity] \thlabel{thm:UBF} Consider the density matrix $\check{\V{\rho}}$, given in \eqref{eqn:int_dense_II}, of a pair of entangled qubits shared by agents via a \ac{TKO} channel.  
Then the optimal fidelity of the kept qubit pair after performing \ac{LOCC} is given by
\begin{align}
F^*= \frac{F^2}{F^2+(1-F)^2\big(\frac{\gamma\delta}{\alpha\beta}\big)^2}
\label{eqn:Fidelity_Upper}
\end{align}
\end{Thm}

\begin{proof} The proof is given in Appendix~\ref{pf_thm:UBF}. \end{proof}

\begin{Remark}[Key channel parameters and the optimal fidelity] 
Equation \eqref{eqn:Fidelity_Upper} describes the optimal fidelity as a function of the parameters of the density matrix. 
To understand how parameters of the channel affect the optimal fidelity $F^*$, one can substitute \eqref{eqn:F_pzeta}--\eqref{eqn:gammadelta_pzeta} into \eqref{eqn:Fidelity_Upper} to obtain
\begin{align}
F^*=\frac{1}{2}+\frac{\sqrt{(1-p)(1-|\eta|^2p)}}{(1-p)+(1-|\eta|^2p)}.\label{eqn:Fidelity_Upper_etap}
\end{align}
By taking derivative of \eqref{eqn:Fidelity_Upper_etap} \ac{w.r.t.} to $p$ and $|\eta|$ respectively,
it can verified that $F^*$ is a decreasing function of  $p$ and an increasing function of $|\eta|$.
An intuitive understanding of such trends can be obtained by recalling Remark~\ref{rem:chanel_para} and Fig.~\ref{fig_channels}.
Operator $\M{C}_2$ destroys all entanglement when it operates on a qubit, and
$p$ is proportional to the probability that $\M{C}_2$ operates.
The larger the $p$, the less entanglement there is after qubits pass through the channel, thereby resulting in a lower $F^*$.
The angle at which $\M{C}_2$ rotates a qubit is given by $\arcsin(|\eta|)$.
The larger $|\eta|$, the easier it is to detect which qubits are operated by $\M{C}_2$, thereby resulting in a higher $F^*$.
In particular, when $|\eta|=1$, $F^*=1$ provided that $p<1$.
Therefore, for amplitude-damping channels, it is possible to design recurrence \ac{QED} algorithms that generate maximally entangled qubit pairs as long as the channel does not completely destroy entanglement.~\QEDB
\end{Remark}

\subsection{Achieving the Optimal Fidelity}
\label{subsec:achfid}
The following algorithm first adapts to the channel so that the prepared qubit pairs have density matrices with a structure invariant to the channel.
%The following algorithm first prepares qubit pairs with density matrices having a fixed structure by exploiting the knowledge of the channel.
Then the algorithm employs recurrent operations to progressively improve the fidelity of the kept qubit pairs.
These operations are specially designed to match the prepared density matrix structure, so that the proposed algorithm achieves optimal fidelity in every round of distillation.

\begin{Algno}[Adaptive recurrence \ac{QED} algorithm]{\color{white}.}\vspace{-4mm}\\
\label{alg:dis}
\begin{itemize}
\item{\bf \ac{RSSP}:} For each qubit pair, the agents transform the density matrix into $\check{\V{\rho}}$ using pre-distillation unitary operators $\M{U}_\mathrm{A}$ and $\M{U}_\mathrm{B}$.\footnote{Given a \ac{TKO} channel, if $\eta=0$, $\M{U}_\mathrm{A}$ and $\M{U}_\mathrm{B}$ are determined by \eqref{eqn:UABphase}. Otherwise,
one can obtain $\V{\rho}$ via \eqref{eqn:initialstate}, then perform \ac{SVD} and Schmidt decomposition sequentially to get \eqref{eqn:psi_Sch} and \eqref{eqn:phi_Sch}, and finally determine $\M{U}_\mathrm{A}$ and $\M{U}_\mathrm{B}$ via \eqref{eqn:UAB}. } Then Bob apply pre-distillation measurement operators 
\begin{align}\V{M}_{\mathrm{B}}=
\mybmatrix{1mm}{2mm}{1}{\kappa&0\\0&1},\qquad \V{M}_{\bar{\mathrm{B}}}=\begin{bmatrix}\sqrt{1-\kappa^2}&0\\0&0\end{bmatrix}\label{eqn:M12}
\end{align}
on his qubit, where $\kappa = \frac{{\beta}}{{\alpha}}$.
If the measurement result corresponds to $\V{M}_{\mathrm{B}}$, Bob performs no further action; otherwise, he notifies  Alice via classical communication and the agents discard the qubit pair.
\item{\bf First round distillation:} The agents take two of the kept qubit pairs, perform the following operations, and repeat these operations on all kept qubit pairs.

(i) Each agent locally performs CNOT operation, i.e., $\M{U}=|00\rangle\langle00|+|01\rangle\langle01|+|10\rangle\langle11|+|11\rangle\langle10|$ on the two qubits at hand.

(ii) Each agent measures the target bit (i.e., the qubit in the second pair) using operators $|0\rangle\langle0|$, $|1\rangle\langle1|$, and transmits the measurement result to the other agent via classical communication.

(iii) If their measurement results do not agree, the agents discard the source qubit pair (i.e., the first pair). If the measurement results agree and correspond to $|1\rangle\langle1|$, the agents keep the source qubit pair. If the measurement results agree and correspond to $|0\rangle\langle0|$, the agents
may choose to discard or keep the source qubit pair; 
the approach that discard or keep the qubit pair in this case is referred to as the \ac{FP} or \ac{PP} approach, respectively.

\item{\bf Following rounds:} Agents perform the same operations as in the first round, except that they always adopt the \ac{PP} approach, i.e., keep the source qubit pair as long as the measurement results agree. Repeat this step until the fidelity of the kept qubit pairs exceeds the required threshold.~\QEDB
\end{itemize}
\end{Algno}

For notational convenience, denote the fidelity of the kept qubit pairs after $n$-th rounds of iteration as ${F}_{n}$, where  $F_{0}=F$. The following theorem characterizes the performance of the proposed algorithm.

\begin{Thm}[Performance of the proposed algorithm]\thlabel{thm:dis_perf}
After the \ac{RSSP} and first round of distillation, a qubit pair is kept with probability
\begin{align}
P_1=\left\{
\begin{array}{l@{\;}l}
\dfrac{F_0^2\alpha^2\beta^4 + (1-F_0)^2\beta^2\gamma^2\delta^2}{2F_0\alpha^2\beta^2 + (1-F_0)(\alpha^2\gamma^2 +\beta^2\delta^2)}&
\begin{array}{l}\mbox{for the \ac{FP}}\\ \mbox{approach}\end{array}\vspace{3mm}\\
\dfrac{4F_0^2\alpha^4\beta^4 + (1-F_0)^2(\alpha^2\gamma^2 +\beta^2\delta^2)^2}{4F_0\alpha^4\beta^2 + 2(1-F_0)\alpha^2(\alpha^2\gamma^2 +\beta^2\delta^2)}&
\begin{array}{l}\mbox{for the \ac{PP}}\\ \mbox{approach}\end{array}
\end{array}
\right.\label{eqn:P_FP1}
\end{align}
and fidelity
\begin{align}
F_1=\left\{\begin{array}{l@{\;}l}
\dfrac{F_0^2}{F_0^2+(1-F_0)^2(\frac{\gamma\delta}{\alpha\beta})^2}
&
\begin{array}{l}\mbox{for the \ac{FP}}\\ \mbox{approach}\end{array}\vspace{1.5mm}\\
\dfrac{F_0^2}{F_0^2+\frac{1}{4}(1-F_0)^2(\frac{\gamma^2}{\beta^2}+\frac{\delta^2}{\alpha^2})^2}
&
\begin{array}{l}\mbox{for the \ac{PP}}\\ \mbox{approach}\end{array}
\end{array}\right.\label{eqn:F_update_FP1}
\end{align}
In the $k$-th round ($k=2,3,4,\ldots$) of distillation, a qubit pair is kept with probability 
\begin{align}P_k=\frac{1}{2}\big(F_{k-1}^2+(1-F_{k-1})^2\big)\end{align} 
and fidelity
\begin{align}
F_{k}=\frac{F_{k-1}^2}{F_{k-1}^2+(1-F_{k-1})^2}.\label{eqn:F_update_PP}
\end{align}
\end{Thm}
\begin{proof}
The proof is given in Appendix~\ref{pf_thm:dis_perf}.
\end{proof}

In the following, algorithms that adopt the \ac{FP} and \ac{PP} approaches in the first round of distillation are referred to as \ac{FP} and \ac{PP} algorithms, respectively.

\begin{Remark}[Convergence speed of fidelity] \label{remark:converge}
For the \ac{FP} algorithm, the density matrix of the kept qubit pair after the $k$-th round of distillation is
\begin{align*}
\V{\rho}^{(k)}=F_{k-1}|\Phi^+\rangle\langle\Phi^+| +(1-F_{k-1})|\Psi^+\rangle\langle\Psi^+|.
\end{align*}
In this case, by comparing \eqref{eqn:Fidelity_Upper} with \eqref{eqn:F_update_FP1} or \eqref{eqn:F_update_PP}, it can be observed that the \ac{FP} algorithm achieves the optimal fidelity in every round of distillation.
This implies that the \ac{FP} algorithm attains the fastest convergence speed \ac{w.r.t.} the rounds of distillation.

The \ac{PP} algorithm achieves a lower fidelity in the first round compared to the \ac{FP} algorithm. 
On the other hand, \eqref{eqn:P_FP1} shows that the probability of keeping a qubit pair in the first round is higher with the \ac{PP} algorithm compared to the \ac{FP} algorithm by a factor more than  2. In particular, when the channel is phase-damping, i.e., $\alpha=\beta=\gamma=\delta=\frac{1}{\sqrt{2}}$, the \ac{PP} algorithm doubles the probability of keeping a qubit pair without lowering the fidelity achieved in the first round.

For the first recurrence \ac{QED} algorithm (will be referred to as BBPSSW algorithm in this work) proposed in \cite{BenBraPopSchSmoWoo:96}, the fidelity of the kept qubit pairs after $k$-th round of distillation is given by
\begin{align}
{F}_k=\frac{F_{k-1}^2+\frac{1}{9}(1-F_{k-1})^2}{F_{k-1}^2+\frac{2}{3}F_{k-1}(1-F_{k-1})+\frac{5}{9}(1-F_{k-1})^2}\,.\label{eqn:speed-old}
\end{align}
Therefore, when $F_0>\frac{1}{2}$, it can be shown using \eqref{eqn:speed-old} that
\begin{align}
\lim_{k\rightarrow \infty}\frac{1-F_k}{1-F_{k-1}}=\frac{2}{3} \,. \label{eqn:speed1}
\end{align}
For the proposed algorithms, when  $F_0>\frac{1}{2}$, it can be shown using \eqref{eqn:F_update_PP} that
\begin{align}
\lim_{k\rightarrow \infty}\frac{1-F_k}{1-F_{k-1}}=0 \,, \qquad
 \lim_{k\rightarrow \infty}\frac{1-F_k}{(1-F_{k-1})^2}=1 \,.\label{eqn:speed2}
\end{align}
Equation
\eqref{eqn:speed1} shows that with the BBPSSW algorithm, the fidelity of the qubit pairs converges to $1$ linearly at rate $\frac{2}{3}$, whereas \eqref{eqn:speed2} shows that with the proposed algorithms, the fidelity converges to $1$ quadratically.
Hence, the convergence speed of the proposed algorithms is significantly improved, i.e., from linear to quadratic, compared to the BBPSSW algorithm.~\QEDB
\end{Remark}

\begin{Remark}[Connection to the \ac{QPA} algorithm]\label{remark:qpa} When the channel is phase-damping, i.e., $\eta=0$, i) the pre-distillation unitary operators  $\M{U}_{\mathrm{A}}=\M{U}_{\mathrm{B}}=\V{H}$ according to \eqref{eqn:UABphase}; ii) the pre-distillation measurement operator $\V{M}_{\mathrm{B}}=\mathbb{I}_2$ since $\alpha=\beta$ according to \eqref{eqn:alphabeta_pzeta}. 
In this case, both local operators employed by Alice and Bob in the \ac{RSSP} are equal to the Hadamard transform $\V{H}$, and hence the \ac{PP} algorithm becomes the \ac{QPA} algorithm in \cite{DeuEkeJozMacPopSan:96}.
Therefore, the \ac{QPA} algorithm is a special case of \ac{PP} algorithm, which employs fixed pre-distillation operators for all channels.
With such non-adaptive pre-distillation operators, the convergence of the fidelity achieved by the \ac{QPA} algorithm is not guaranteed \cite{DeuEkeJozMacPopSan:96}.
With the proposed adaptive pre-distillation operators, the fidelity achieved by both \ac{FP} and \ac{PP} algorithms converges quadratically for \ac{TKO} channels.
The proposed algorithms may be applied to more general channels, yet their convergence property for such channels remains to be characterzied.~\QEDB
\end{Remark}

\begin{Remark}[Benefit of channel adaptation] In the proposed algorithms, the channel adaptation takes place in the \ac{RSSP}.
As shown in \thref{thm:dis_perf} and Remark~\ref{remark:converge}, 
despite its simplicity of involving single-qubit operations only in the initial step, 
\ac{RSSP} is the keystone to improve the effectiveness of distillation for \ac{TKO} channels.
With the BBPSSW algorithm \cite{BenBraPopSchSmoWoo:96}, in addition to the distillation operations, random bilateral rotations are required to restore the desired density matrix structure in every round of distillation.
With the \ac{QPA} algorithm, no random rotations are required, yet the density matrix structure may not be preserved for different rounds of distillation.
In the proposed algorithms, the \ac{RSSP} adapts to the channel  so that the prepared qubit pairs have density matrices with a structure invariant to the channel.
As a result, the distillation operation itself, which involves only the CNOT operation and single-qubit measurements, is sufficient to maintain the density matrix structure in every round of distillation. 
This feature enables simple \ac{QED} algorithm with guaranteed convergence.
Hence, channel adaptation also improves the implementability of \ac{QED} algorithms.~\QEDB
\end{Remark}

\section{Numerical Results}
This section provides numerical results to demonstrate the performance of the proposed algorithms. In particular, the proposed \ac{FP} and \ac{PP} algorithms are compared with the BBPSSW algorithm in \cite{BenBraPopSchSmoWoo:96} and the \ac{QPA} in \cite{DeuEkeJozMacPopSan:96} for a required fidelity $F_{\mathrm{th}}=0.99$.

\begin{figure*}[t] \centering
\psfrag{0.5}[Br][Br][0.6]{0.5\hspace{0.3mm}}
\psfrag{0.6}[Br][Br][0.6]{0.6\hspace{0.3mm}}
\psfrag{0.7}[Br][Br][0.6]{0.7\hspace{0.3mm}}
\psfrag{0.8}[Br][Br][0.6]{0.8\hspace{0.3mm}}
\psfrag{0.9}[Br][Br][0.6]{0.9\hspace{0.3mm}}
\psfrag{1}[Br][Br][0.6]{1\hspace{0.3mm}}
\psfrag{0}[tt][tt][0.6]{0}
\psfrag{4}[tt][tt][0.6]{4}
\psfrag{5}[tt][tt][0.6]{5}
\psfrag{8}[tt][tt][0.6]{8}
\psfrag{10}[tt][tt][0.6]{10}
\psfrag{12}[tt][tt][0.6]{12}
\psfrag{15}[tt][tt][0.6]{15}
\psfrag{16}[tt][tt][0.6]{16}
\psfrag{24}[tt][tt][0.6]{24}
\psfrag{Rounds of distillation}[tc][tc][0.7]{Rounds of distillation}
\psfrag{Fidelity}[tc][tc][0.7]{Fidelity}
\psfrag{Classical           Algorithm}[cl][cl][0.7]{\hspace{-2.5mm}BBPSSW Algorithm}
\psfrag{FP      Algorithm}[cl][cl][0.7]{\hspace{-2.5mm}FP Algorithm}
\psfrag{PP      Algorithm}[cl][cl][0.7]{\hspace{-2.5mm}PP Algorithm}
\psfrag{QPA      Algorithm}[cl][cl][0.7]{\hspace{-2.5mm}QPA Algorithm}
\psfrag{P}[cc][cc][0.7]{Phase-damping Channel}
\psfrag{M}[cc][cc][0.7]{``Mid-point" Channel}
\psfrag{A}[cc][cc][0.7]{Amplitude-damping Channel}
\hspace{-10mm}
\includegraphics[scale=0.57]{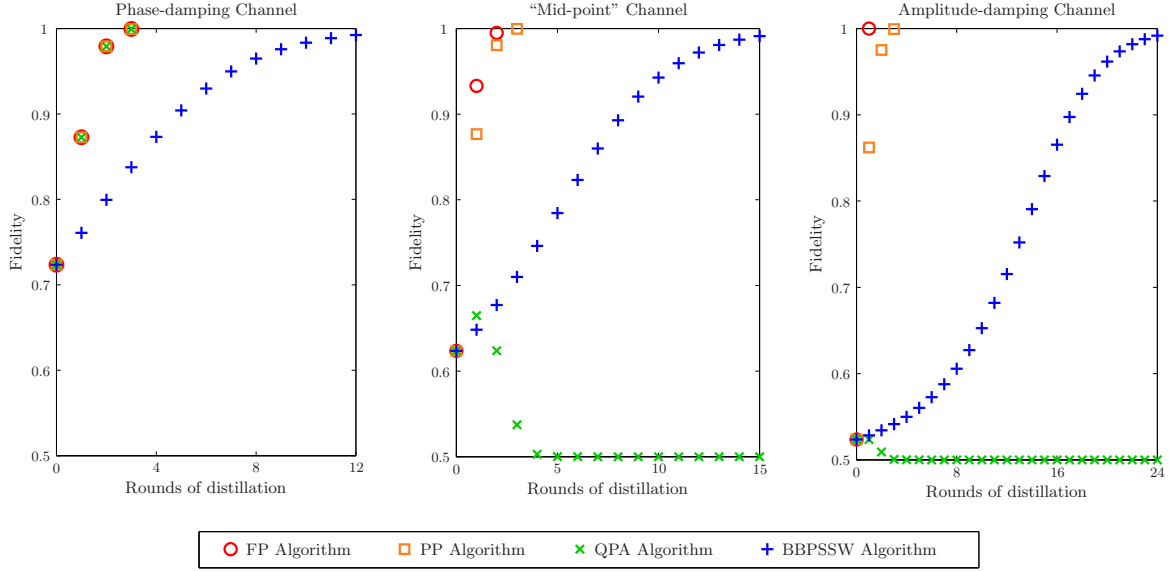}
%\vspace{-1mm}
\caption{\small The achieved fidelity as a function of the rounds of distillation for a phase-damping channel, a ``mid-point" channel, and an amplitude-damping channel. For all channels, the noise severity parameter $p=0.8$.}
\label{fig_speed}%
\end{figure*}

Fig.~\ref{fig_speed} shows the fidelity of kept qubit pairs as a function of the rounds of distillation for three types of channels, i.e., 
a phase-damping channel, a ``mid-point"  channel\footnote{This channel ($\arcsin(|\eta|)=\frac{\pi}{4}$) can be thought of as the mid-point of phase-damping channels ($\arcsin(|\eta|) = 0$) and amplitude-damping channels ($\arcsin(|\eta|)=\frac{\pi}{2}$).}, and an amplitude-damping channel.
When the channel is phase-damping, the fidelity achieved by \ac{FP}, \ac{PP}, and \ac{QPA} algorithms are the same, which is consistent with the observations made in Remark~\ref{remark:converge} and~\ref{remark:qpa}.
When the channel is ``mid-point" or amplitude-damping, the \ac{QPA} algorithm does not achieve the required fidelity, illustrating its converge issue.
The \ac{FP}, \ac{PP}, and the BBPSSW algorithm achieve the required fidelity on all channels, with the proposed algorithms using much less rounds of distillation.
For instance, when the channel is amplitude-damping, the BBPSSW algorithm requires 24 rounds of distillation, whereas the \ac{FP} and \ac{PP} algorithms only require one and three rounds respectively. 
Since the yield is reduced by at least half after each round of distillation, the yield of the proposed algorithms are significantly higher than the classical one for all the considered channel.

\begin{figure}[t] \centering
\psfrag{1-0}[Br][Br][0.6]{$10^0$\hspace{2.1mm}}
\psfrag{1-1}[Br][Br][0.6]{$10^{-1}$\hspace{-0.3mm}}
\psfrag{1-2}[Br][Br][0.6]{$10^{-2}$\hspace{0.3mm}}
\psfrag{1-3}[Br][Br][0.6]{$10^{-3}$\hspace{0.3mm}}
\psfrag{1-4}[Br][Br][0.6]{$10^{-4}$\hspace{0.3mm}}
%\psfrag{1-5}[Br][Br][0.6]{$10^{-5}$\hspace{0.3mm}}
%\psfrag{1-6}[Br][Br][0.6]{$10^{-6}$\hspace{0.3mm}}
\psfrag{0}[tt][tt][0.6]{0.0}
\psfrag{0.1}[tt][tt][0.6]{0.1}
\psfrag{0.2}[tt][tt][0.6]{0.2}
\psfrag{0.3}[tt][tt][0.6]{0.3}
\psfrag{0.4}[tt][tt][0.6]{0.4}
\psfrag{0.5}[tt][tt][0.6]{0.5}
\psfrag{0.6}[tt][tt][0.6]{0.6}
\psfrag{0.7}[tt][tt][0.6]{0.7}
\psfrag{0.8}[tt][tt][0.6]{0.8}
\psfrag{0.9}[tt][tt][0.6]{0.9}
\psfrag{1}[tt][tt][0.6]{1.0}
\psfrag{p}[cc][cc][0.75]{$p$}
\psfrag{Yield}[tc][tc][0.7]{Yield}
\psfrag{Classical           Algorithm}[cl][cl][0.6]{\hspace{1mm}BBPSSW algorithm}
\psfrag{FP}[cl][cl][0.6]{\hspace{1mm}\ac{FP} algorithm}
\psfrag{PP}[cl][cl][0.6]{\hspace{1mm}\ac{PP} algorithm}
\includegraphics[scale=0.43]{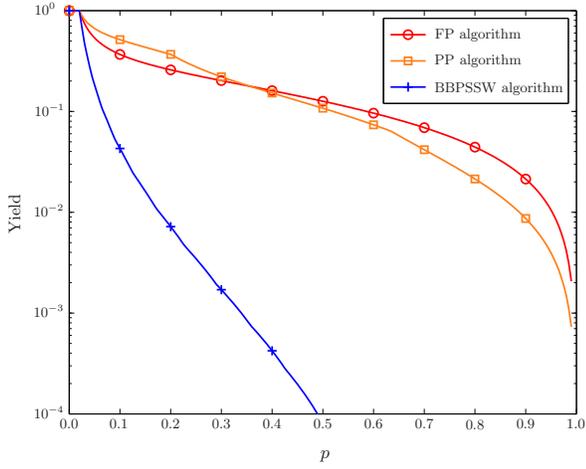}
%\vspace{-1mm}
\caption{\small The efficiency of different recurrence \ac{QED} algorithms
as a function of noise severity $p$ for amplitude-damping channels ($p\in[0,0.99]$,
$\eta = 1$). The \ac{QPA} algorithm is not plotted as it does not achieve
the required fidelity on amplitude-damping channels.}
\label{fig_P}
\end{figure}

Fig.~\ref{fig_P} shows the yield\footnote{Suppose the fidelity $F_k$ exceeds the required threshold $F_{\mathrm{th}}$ after $K$ rounds of distillation. 
In the following,  the agents are assumed to generate qubit pairs which go through $K-1$ and $K$ rounds of distillation with probability $\frac{F_{K}-F_{\mathrm{th}}}{F_{K}-F_{K-1}}$ and $\frac{F_{\mathrm{th}}-F_{K-1}}{F_{K}-F_{K-1}}$, respectively. This assumption assures that the average output fidelity of the algorithms is always $F_{\mathrm{th}}$, enabling a fair comparison among different scenarios. Denote the yield after $K-1$ and $K$ rounds of distillation as $Y_{K-1}$ and $Y_k$, respectively, then the average yield of the algorithm is given by $\frac{F_{K}-F_{\mathrm{th}}}{F_{K}-F_{K-1}}Y_{k-1}+\frac{F_{\mathrm{th}}-F_{K-1}}{F_{K}-F_{K-1}}Y_k$.} of the distillation algorithms as a function of the noise severity $p$. 
While the yield of all algorithms decreases with increasing $p$, the proposed algorithms are much more resilient to noise compared to the BBPSSW algorithm. Comparing the two proposed algorithms, the \ac{FP} algorithm performs better for large $p$, whereas the \ac{PP} algorithm performs better for small $p$. This shows that when the noise is severe, it is beneficial to the increase the achieved fidelity at a cost of reducing the probability of keeping qubit pairs.

Finally, Fig.~\ref{fig_eta} shows the yield of the distillation algorithms as a function of the channel type parameter $|\eta|$. 
The \ac{QPA} algorithm has the same efficiency as the \ac{PP} algorithm when the channel is phase-damping, 
which is consistent with Remark~\ref{remark:qpa}.
Yet the \ac{QPA} algorithm does not achieve the required fidelity when $\arcsin{|\eta|}\ge 0.024\pi$.
This illustrates the importance of channel adaptation.
Comparing the two proposed algorithms, the \ac{FP} algorithm is more efficient when the channel tends towards an amplitude-damping channel (i.e., $|\eta|$ approaches 1), and the \ac{PP} algorithm is more efficient when the channel tends towards a phase-damping channel (i.e., $|\eta|$ approaches 0). This is consistent with \thref{thm:Str_rho} and \thref{thm:dis_perf}, which show that the benefit of increasing the fidelity by adopting the \ac{FP} algorithm is greater when $\eta$ is close to 1 and vice versa. 
It can also be seen that the yield of the \ac{PP} algorithm is twice of the \ac{FP} algorithm when $\eta=0$.
This is consistent with the observation made in Remark~\ref{remark:converge} that the \ac{PP} algorithm doubles the probability of keeping a qubit pair without lowering the fidelity achieved in the first round when the channel is phase-damping.
\begin{figure}[t] \centering
\psfrag{0.0}[Br][Br][0.6]{0\hspace{0.3mm}}
\psfrag{0.02}[Br][Br][0.6]{0.02\hspace{0.3mm}}
\psfrag{0.04}[Br][Br][0.6]{0.04\hspace{0.3mm}}
\psfrag{0.06}[Br][Br][0.6]{0.06\hspace{0.3mm}}
\psfrag{0.08}[Br][Br][0.6]{0.08\hspace{0.3mm}}
\psfrag{0.10}[Br][Br][0.6]{0.10\hspace{0.3mm}}
\psfrag{0.12}[Br][Br][0.6]{0.12\hspace{0.3mm}}
\psfrag{0.14}[Br][Br][0.6]{0.14\hspace{0.3mm}}
\psfrag{0.16}[Br][Br][0.6]{0.16\hspace{0.3mm}}
\psfrag{0}[tt][tt][0.6]{0}
\psfrag{0.1}[tt][tt][0.6]{}
\psfrag{0.2}[tt][tt][0.6]{$0.1\pi$}
\psfrag{0.3}[tt][tt][0.6]{}
\psfrag{0.4}[tt][tt][0.6]{$0.2\pi$}
\psfrag{0.5}[tt][tt][0.6]{}
\psfrag{0.6}[tt][tt][0.6]{$0.3\pi$}
\psfrag{0.7}[tt][tt][0.6]{}
\psfrag{0.8}[tt][tt][0.6]{$0.4\pi$}
\psfrag{0.9}[tt][tt][0.6]{}
\psfrag{1}[tt][tt][0.6]{$0.5\pi$}
\psfrag{asineta}[tc][tc][0.75]{$\arcsin(|\eta|)$}
\psfrag{Yield}[tc][tc][0.7]{Yield}
\psfrag{FP}[cl][cl][0.6]{\hspace{0.5mm}FP algorithm}
\psfrag{PP}[cl][cl][0.6]{\hspace{0.5mm}PP algorithm}
\psfrag{QPA Algorithm}[cl][cl][0.6]{\hspace{0.5mm}\ac{QPA} algorithm}
\psfrag{Classical           Algorithm}[cl][cl][0.6]{\hspace{0.5mm}BBPSSW algorithm}
\includegraphics[scale=0.43]{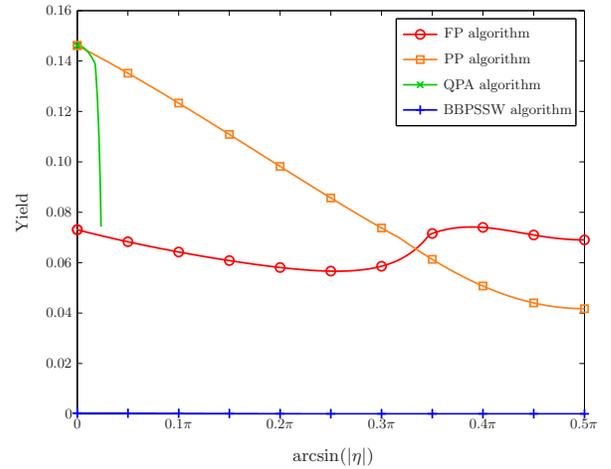}
%\vspace{-1mm}
\caption{\small The efficiency of the two proposed \ac{QED} algorithms as a function of channel type parameter $|\eta|$ ($p = 0.7$, $|\eta|\in[0,1]$). The BBPSSWBBPSSW algorithm algorithm is not plotted as its yield is between $2.9\times10^{-4}$ and $3.8\times 10^{-7}$ on the considered channels.}
\label{fig_eta}%
\end{figure}

%----------------------------------------------------------------------------%
%                              Conclusion                                   %
%---------------------------------------------------------------------------%
\section{Conclusion}\label{sec:conclude}
Among various types of \ac{QED} algorithms, the recurrence ones require quantum operations on the minimum number of qubits and can generate maximally entangled qubit pairs even when the noise in the channel is severe.
Despite their advantages, the efficiency issue of recurrence \ac{QED} algorithms has not been thoroughly investigated in the literature.
In this paper, we first characterize the effect of a single-qubit \ac{TKO} channel on the entanglement of a qubit pair shared by the agents via this channel.
We then determine the optimal fidelity that can be achieved by performing \ac{LOCC} on two of such qubit pairs.
Finally we propose two adaptive recurrence \ac{QED} algorithms, one of which achieves the optimal fidelity.
The proposed algorithms preserve the density matrix structure in every round of distillation, avoiding the need of additional random rotations.
This enable simple \ac{QED} algorithms with guaranteed convergence for \ac{TKO} channels.
In fact, the convergence speed of both algorithms are improved from linear to quadratic compared to the BBPSSW algorithm.
Numerical results confirm that the proposed algorithms significantly improve the efficiency of recurrence \ac{QED} algorithms.
These results also indicate that the benefit of achieving the optimal fidelity is greater when the noise is severe, or the channel tends towards an amplitude-damping channel.

\section*{Acknowledgement}
The authors would like to thank Aram Harrow and Peter W. Shor for the valuable discussions and suggestions.

%---------------------------------------------------------------------------%
%                               Appendices                                  %
%---------------------------------------------------------------------------%
\appendix
\section{Proof of \thref{lem:channel}}
\label{pf_lem:channel}
Consider a single-qubit \ac{TKO} channel represented by $\M{C}_1$, $\M{C}_2$. We
will first prove the theorem for the case in which $\mathrm{rank}\{\M{C}_2\}=1$, then show that the case in which $\mathrm{rank}\{\M{C}_2\}=2$
can be transformed into the prior case.

When $\mathrm{rank}\{\M{C}_2\}=1$,  \ac{SVD} of $\M{C}_2$ shows that there exists $|i\rangle, |j\rangle \in \mathbb{C}^2$, $p\in(0,1]$, and $\epsilon\in\mathbb{R}$ such that
\begin{align}
\M{C}_2=\sqrt{p}e^{\imath\epsilon}|i\rangle\langle j|.\label{eqn:C2}
\end{align}
Noting that a quantum operator is invariant up to an overall phase change, $\epsilon$ can be any real number.
Recall from \eqref{eqn:channel_identity} that
\begin{align}
\M{C}^\dag_1\M{C}_1=\mathbb{I}_2-\M{C}^\dag_2\M{C}_2.\label{eqn:C1C1dag}
\end{align}
Substituting the \ac{SVD} of $\M{C}_1=\M{U}_\mathrm{c}\V{D}_\mathrm{c}\V{V}^\dag_\mathrm{c}$ and \eqref{eqn:C2}  into \eqref{eqn:C1C1dag}, one can get 
\begin{align}
\V{V}_\mathrm{c}\V{D}^2_\mathrm{c}\V{V}^\dag_\mathrm{c} = |\tilde{j}\rangle\langle\tilde{j}| + (1-p)|j\rangle\langle j|\label{eqn:VDV}
\end{align}
where $ \langle j|\tilde{j}\rangle=0$. Since $\V{D}^2_\mathrm{c}$ is diagonal and $\V{V}_\mathrm{c}$ is unitary, 
\begin{align*}\V{D}_\mathrm{c}=\begin{bmatrix}1&0\\0&\sqrt{1-p}\end{bmatrix}, \quad\mbox{and}\quad \V{V}_\mathrm{c}=\big[|\tilde{j}\rangle \; |j\rangle\big].
\end{align*}
Hence, there exists $|\tilde{k}\rangle$ and $|k\rangle$ with $\langle \tilde{k}| k\rangle =0$ such that
\begin{align}
\M{C}_1=|\tilde{k}\rangle\langle \tilde{j}|+\sqrt{1-p}|k\rangle\langle\label{eqn:C1}
j|.
\end{align}

It can be verified that $\M{C}_1$ in \eqref{eqn:C1} and and $\M{C}_2$ in \eqref{eqn:C2} can be expressed in the form given in \eqref{eqn:channel_structure}, with
\begin{align}
\begin{split}
\M{U}&=|\tilde{k}\rangle\langle 0|+|k\rangle\langle 1|\\
\V{V}&=|\tilde{j}\rangle\langle 0|+|j\rangle\langle 1|\\
\eta&=e^{\imath\epsilon}\langle \tilde{k} |i\rangle\\
\zeta&=e^{\imath\epsilon}\langle k |i\rangle\\
\epsilon &= -\mathrm{pha}\{\langle k |i\rangle\}.
\end{split}\label{eqn:UVeps}
\end{align}
This completes the proof for the case with  $\mathrm{rank}\{\M{C}_2\}=1$.

Now consider the case in which  $\mathrm{rank}\{\M{C}_2\}=2$. Since $\M{C}_2$
is full rank, $\det\{\M{C}_2\}\neq 0$.  Consider equation
\begin{align}
\det\{-\M{C}_1+x\M{C}_2\}=0.\label{eqn:det0}
\end{align}
This is a second-order polynomial equation of $x$, for which the coefficient
of the second-order term is $\det\{\M{C}_2\}\neq 0$. Therefore, the
fundamental theorem of algebra implies that \eqref{eqn:det0} must have at least one solution. Denote
$x_0$ as one of the solutions of \eqref{eqn:det0}.
Recall from \cite[Sec 3.3]{Pre:B}, any single-qubit \ac{TKO} channel with operators $\{\M{C}_k\}$ can be equivalently represented by operators $\{\tilde{\M{C}}_k\}$ satisfying
\begin{align}
\big[\tilde{\M{C}}_1 \; \tilde{\M{C}}_2\big]
=\big[\M{C}_1  \; \M{C}_2\big](\V{A}
\otimes\mathbb{I}_{2})\label{eqn:tildeC}
\end{align}
where $\V{A}$
is an arbitrary unitary matrix.
In particular, let
\begin{align*}
\V{A}=\frac{1}{\sqrt{1+|x_0|^2}}\begin{bmatrix}x_0^\dag&-1\\1&x_0\end{bmatrix}\end{align*}
then $\det\{\tilde{\M{C}}_2\}=\dfrac{\det\{-\M{C}_1+x_0\M{C}_2\}}{1+|x_0|^2}=0$.
Thus $\mathrm{rank}\{\tilde{\M{C}}_2\}\le1$. If it were the case that $\mathrm{rank}\{\tilde{\M{C}}_2\}=0$, then $\M{C}_2=\V{0}$ implying that the channel has only one operator. This contradicts that the channel has two operators. Hence, $\mathrm{rank}\{\tilde{\M{C}}_2\}=1$,
which is the case that has been proven above.

\section{Proof of \thref{thm:Str_rho}}
\label{pf_thm:Str_rho}
Since $\mathrm{rank}\{\V{\rho}_0\}=1$, and the channel has only two operators,
from \eqref{eqn:initialstate}, $\mathrm{rank}\{\V{\rho}\}\le2$. 
Since density matrices are Hermitian, the spectral decomposition gives
\begin{align}
\V{\rho}=F|\psi\rangle\langle\psi| + (1-F)|\phi\rangle\langle\phi|\label{eqn:int_dense_I}
\end{align}
where
\begin{align}
\langle\psi|\phi\rangle=0.\label{eqn:phipsiorth}
\end{align}
In \eqref{eqn:int_dense_I}, we have used the fact that density matrices have trace 1. Without loss of generality, assume $F\in[\frac{1}{2},1]$.

If $p = 0$, \eqref{eqn:initialstate} and \eqref{eqn:channel_structure2} imply that $\V{\rho}=\V{\rho}_0$. Setting $\M{U}_{\mathrm{A}}=\M{U}_{\mathrm{B}}=\mathbb{I}_2$ in \eqref{eqn:int_dense_II}, it is straightforward that the theorem holds.
Also, if $\eta = 0$, the channel is phase-damping. Then
\begin{align*}
\V{\rho} &= \frac{1}{2}\big( (|00\rangle + \sqrt{1-p}|11\rangle)(\langle00| + \sqrt{1-p}\langle11|) + p|11\rangle\langle11| \big)\\
&= \frac{1}{2}\big( |00\rangle\langle00| + \sqrt{1-p} |00\rangle\langle11| + \sqrt{1-p} |11\rangle\langle00| +  |11\rangle\langle11| \big)\\
&= F|\psi\rangle\langle\psi| +(1-F)|\phi\rangle\langle\phi|
\end{align*}
where
\begin{align*}
F &=\frac{1+\sqrt{1-p}}{2}\\
|\psi\rangle &=\frac{1}{\sqrt{2}}(|00\rangle+|11\rangle)\\
|\phi\rangle &=\frac{1}{\sqrt{2}}(|00\rangle-|11\rangle).
\end{align*}
Setting both local unitary operators  in \eqref{eqn:int_dense_II} to be Hadamard transform, i.e.,
\begin{align}
\M{U}_{\mathrm{A}}=\M{U}_{\mathrm{B}}=\V{H}=\frac{1}{\sqrt{2}}\begin{bmatrix}1&1\\1&-1\end{bmatrix}\label{eqn:UABphase}
\end{align}
it is easy to see that the theorem also holds for the case of $\eta = 0$.
Therefore, the following analysis consider the case for which $p \in (0,1)$ and $|\eta|>0$.

We first determine the value of $F$. Set $\V{A}=\begin{bmatrix}\kappa^\dag&-\lambda^\dag\\\lambda&\kappa\end{bmatrix}$ in \eqref{eqn:tildeC} with $\kappa, \lambda\in\mathbb{C}$ such that $|\kappa|^2+|\lambda|^2=1$. Then, it can be shown that
\begin{align}
\begin{split}
\V{\rho}&\stackrel{(a)}{=} \sum_{k=1}^{2}(\mathbb{I}_2\otimes\tilde{\M{C}}_k)|\Phi^+\rangle\langle\Phi^+|(\mathbb{I}_2\otimes\tilde{\M{C}}_k)^\dag\\
&\stackrel{(b)}{=}\big(\mathbb{I}_2\otimes(\kappa^\dag\M{C}_1+\lambda\M{C}_2)\big)|\Phi^+\rangle\langle\Phi^+|
\big(\mathbb{I}_2\otimes(\kappa^\dag\M{C}_1+\lambda\M{C}_2)\big)^\dag
\\&+\big(\mathbb{I}_2\otimes(-\lambda^\dag\M{C}_1+\kappa\M{C}_2)\big)|\Phi^+\rangle\langle\Phi^+|
\big(\mathbb{I}_2\otimes(-\lambda^\dag\M{C}_1+\kappa\M{C}_2)\big)^\dag\\
&\stackrel{(c)}{=}\frac{1}{2}\big(\V{v}_1\V{v}^\dag_1+\V{v}_2\V{v}^\dag_2\big),
\end{split}
\label{eqn:rho_abc}
\end{align}
where
\begin{align*}
\V{v}_1=\begin{bmatrix}\kappa^\dag\\ 0\\\lambda\eta\sqrt{p}\\\kappa^\dag\sqrt{1-p}+\lambda\zeta\sqrt{p}\end{bmatrix},\;\;
\V{v}_2=\begin{bmatrix}-\lambda^\dag\\ 0\\\kappa\eta\sqrt{p}\\-\lambda^\dag\sqrt{1-p}+\kappa\zeta\sqrt{p}\end{bmatrix}.
\end{align*}
In \eqref{eqn:rho_abc}, (a) is due to \eqref{eqn:initialstate} together with the equivalence between $\{\M{C}_k\}$ and $\{\tilde{\M{C}}_k\}$, (b) is due to \eqref{eqn:tildeC}, and (c) is due to \eqref{eqn:channel_structure2}.

In order to make last line of \eqref{eqn:rho_abc} a spectral decomposition of $\V{\rho}$, it is necessary to make $\V{v}_1$ and $\V{v}_2$ orthogonal. A sufficient condition for $\V{v}^\dag_1\V{v}_2=0$ is given by 
\begin{align*}
\kappa = \bigg(\frac{\sqrt{\frac{1-p}{1-|\eta|^2p}}+1}{2}\bigg)^\frac{1}{2},\qquad \lambda = \sqrt{1-\kappa^2}.
\end{align*}
Seting $|\psi\rangle = \frac{\V{v}_1}{||\V{v}_1||}$ and $|\phi\rangle = \frac{\V{v}_2}{||\V{v}_2||}$ in the last line of \eqref{eqn:rho_abc}, one can get
\begin{align}
\V{\rho}=\frac{1}{2}||\V{v}_1||^2|\psi\rangle\langle\psi| + \frac{1}{2}||\V{v}_2||^2|\phi\rangle\langle\phi| \label{eqn:int_dense_spe}
\end{align}
which is a spectral decomposition of $\V{\rho}$. Since the spectrum of a matrix is unique, by comparing  \eqref{eqn:int_dense_spe}  with \eqref{eqn:int_dense_I}, one gets
\begin{align}
F& = \frac{1}{2}||\V{v}_1||^2\nonumber
\\&=\frac{1}{2}+\frac{1}{2}\sqrt{(1-p)(1-|\eta|^2p)}\label{eqn:F_pzeta2}
\end{align}
which proves \eqref{eqn:F_pzeta}.

We next show \eqref{eqn:mu}, \eqref{eqn:nu}.
Perform Schmidt decomposition on the eigenvectors of $\V{\rho}$ as
\begin{align}
|\psi\rangle &= \alpha|wx\rangle+\beta|\tilde{w}\tilde{x}\rangle,\label{eqn:psi_Sch}\\
|\phi\rangle &= \gamma|yz\rangle +\delta|\tilde{y}\tilde{z}\rangle,\label{eqn:phi_Sch}
\end{align}
where $\langle s|\tilde{s}\rangle = 0$ for $s\in\{w,x,y,z\}$, and $\alpha,\beta,\gamma,\delta\in[0,1]$ satisfying
\begin{align}\alpha^2 +\beta^2 =\gamma^2
+\delta^2=1.\label{eqn:alpha_normal} \end{align}
 Without
loss of generality, assume $\beta\le\alpha$, $\gamma\le \delta$.

From \eqref{eqn:F_pzeta2},  when $p\in(0,1)$, $F\in(\frac{1}{2},1)$. Substituting \eqref{eqn:channel_structure2} into \eqref{eqn:initialstate} and taking partial trace over different qubits, one can obtain
the density matrices of the first and second qubits, i.e.,
\begin{align}
\begin{split}
\V{\rho}_1&=\mathrm{tr}_{2}\{\V{\rho}\}=\frac{1}{2}\mybmatrix{1mm}{2mm}{1}{1&0\\0&1}\\
\V{\rho}_2&=\mathrm{tr}_{1}\{\V{\rho}\}=\frac{1}{2}\begin{bmatrix}1+|\eta|^2p&\eta\zeta p\\\eta^\dag\zeta p&1-p+\zeta^2p\end{bmatrix}.
\end{split}\label{eqn:subden}
\end{align}
On the other hand, substituting \eqref{eqn:psi_Sch} and \eqref{eqn:phi_Sch} into \eqref{eqn:int_dense_I} and taking partial trace, one can obtain alternative expression of $\V{\rho}_1$ and $\V{\rho}_2$ in terms of $|x\rangle$, $x\in\{a,b,c,d\}$. This together with
\eqref{eqn:subden} give
\begin{align}
&F\big(\alpha^2|w\rangle\langle w| + \beta^2 |\tilde{w}\rangle\langle\tilde{w}|\big)+
(1-F)\big(\gamma^2|y\rangle\langle y| + \delta^2 |\tilde{y}\rangle\langle\tilde{y}|\big)\nonumber\\
&= \frac{1}{2}
\mybmatrix{1mm}{2mm}{1}{1&0\\0&1},
\label{eqn:psiphi_1}\\
&F\big(\alpha^2|x\rangle\langle x| + \beta^2 |\tilde{x}\rangle\langle\tilde{x}|\big)+
(1-F)\big(\gamma^2|z\rangle\langle z| + \delta^2 |\tilde{z}\rangle\langle\tilde{z}|\big)\nonumber\\
&=\frac{1}{2}\begin{bmatrix}1+|\eta|^2p&\eta\zeta p\\\eta^\dag\zeta p&1-p+\zeta^2p\end{bmatrix}.\label{eqn:psiphi_2}
\end{align}

We claim that $\gamma<\delta$ when $|\eta|>0$. 
If it were not the case, then $\gamma = \delta=\frac{1}{\sqrt{2}}$. Thus
\begin{align}\gamma^2|y\rangle\langle y| + \delta^2 |\tilde{y}\rangle\langle\tilde{y}|=\frac{1}{2}\mybmatrix{1mm}{2mm}{1}{1&0\\0&1}
\label{eqn:CI}
\end{align}
because  $\langle c|\tilde{c}\rangle = 0$.
Substitute \eqref{eqn:CI} into the left side of \eqref{eqn:psiphi_1}, one can get
\begin{align}
\alpha^2|w\rangle\langle w| + \beta^2 |\tilde{w}\rangle\langle\tilde{w}|=\frac{1}{2}\mybmatrix{1mm}{2mm}{1}{1&0\\0&1}.\label{eqn:AI}
\end{align}
Since  $\langle w|\tilde{w}\rangle = 0$, the left side of \eqref{eqn:AI} is a spectral decomposition of the right side, implying  $\alpha = \beta=\frac{1}{\sqrt{2}}$. Substitute $\alpha = \beta=\gamma = \delta=\frac{1}{\sqrt{2}}$ into the left side of \eqref{eqn:psiphi_2}, and since $\langle x|\tilde{x}\rangle = \langle z|\tilde{z}\rangle=0$, one can get
\begin{align}
\mybmatrix{1mm}{2mm}{1}
{
1&0\\
0&1
}
=
\mybmatrix{0.5mm}{1mm}{1}
{
1+|\eta|^2p&\eta\zeta p\\
\eta^\dag\zeta p&1-p+\zeta^2p
}
\end{align}
which holds only if $|\eta| = 0$. This contradicts with the fact that $|\eta|>0$ and thus proves the claim.

Next construct two unitary operators as
\begin{align}
\begin{split}
\M{U}_{\mathrm{A}}&=|0\rangle\langle w| + |1\rangle\langle \tilde{w}|\\
\M{U}_{\mathrm{B}}&=|0\rangle\langle x| + |1\rangle\langle \tilde{x}|.
\end{split}\label{eqn:UAB}
\end{align}
Substituting this into \eqref{eqn:int_dense_I}, it can be obtained that 
\begin{align}
\check{\V{\rho}}=F|\mu\rangle\langle\mu|
+ (1-F)|\nu \rangle\langle\nu|\label{eqn:checkrho}
\end{align}
in which
\begin{align}
|\mu\rangle &=\alpha|00\rangle + \beta |11\rangle\label{eqn:mu2}\\
|\nu\rangle &=\gamma|y_\mathrm{r}z_\mathrm{r}\rangle + \delta |\tilde{y}_\mathrm{r}\tilde{z}_\mathrm{r}\rangle\label{eqn:nu2}
\end{align}
where the ket notations with subscript ``$\mathrm{r}$" denote the rotated version of the original ones, e.g., 
$|y_\mathrm{r}\rangle=\M{U}_\mathrm{A}|y\rangle$, $|\tilde{z}_\mathrm{r}\rangle=\M{U}_\mathrm{B}|\tilde{z}\rangle$. Equation \eqref{eqn:checkrho} gives the structure of \eqref{eqn:int_dense_II}, and
\eqref{eqn:mu2} proves \eqref{eqn:mu}. 

The following analysis focuses on proving \eqref{eqn:nu}. 
Since $\M{U}_{\mathrm{A}}$, $\M{U}_{\mathrm{B}}$ are unitary,  \eqref{eqn:phipsiorth} implies $\langle\mu|\nu\rangle=0$, which gives
\begin{align}
&\alpha\gamma\langle00|y_\mathrm{r}z_\mathrm{r}\rangle+\alpha\delta\langle00|\tilde{y}_\mathrm{r}\tilde{z}_\mathrm{r}\rangle\nonumber\\
&+\beta\gamma\langle11|y_\mathrm{r}z_\mathrm{r}\rangle+\beta\delta\langle11|\tilde{y}_\mathrm{r}\tilde{z}_\mathrm{r}\rangle=0.\label{eqn:munuorth2}
\end{align}
Substituting $|w\rangle = \M{U}^\dag_{\mathrm{A}}|0\rangle$,  $|\tilde{w}\rangle = \M{U}^\dag_{\mathrm{A}}|1\rangle$, $|y\rangle = \M{U}^\dag_{\mathrm{A}}|y_\mathrm{r}\rangle$, and $|\tilde{y}\rangle = \M{U}^\dag_{\mathrm{A}}|\tilde{y}_\mathrm{r}\rangle$ into \eqref{eqn:psiphi_1}, one can get
\begin{align}
&F(\alpha^2|0\rangle\langle 0| + \beta^2 |1\rangle\langle1|)+
(1-F)(\gamma^2|y_\mathrm{r}\rangle\langle y_\mathrm{r}| + \delta^2 |\tilde{y}_\mathrm{r}\rangle\langle\tilde{y}_\mathrm{r}|)\nonumber
\\&=\frac{1}{2}(|0\rangle\langle0|+|1\rangle\langle1|).
\label{eqn:munu_I2}
\end{align}

Since $\M{U}_{\mathrm{A}}$, $\M{U}_{\mathrm{B}}$ are unitary and $\langle s|\tilde{s}\rangle = 0$, 
$\langle s_\mathrm{r}|\tilde{s}_\mathrm{r}\rangle = 0$, where $s\in\{y,z\}$.
Since $|0\rangle$, $|1\rangle$ and $|y_\mathrm{r}\rangle$, $|\tilde{y}_\mathrm{r}\rangle$ are two
sets of orthonormal basis for two-dimensional Hilbert space, there exists
$a, b\in\mathbb{C}$, $|a|^2 +|b|^2=1$, such that
\begin{align}
|y_\mathrm{r}\rangle=a|0\rangle +b|1\rangle, \quad\mbox{and}\quad |\tilde{y}_\mathrm{r}\rangle=-b^\dag|0\rangle
+a^\dag|1\rangle\label{eqn:c01}
\end{align}
Substitute \eqref{eqn:c01} into \eqref{eqn:munu_I2}, then
\begin{align}
&\Big(F\alpha^2+(1-F)(\gamma^2|a|^2+\delta^2|b|^2)-\frac{1}{2}\Big)|0\rangle\langle
0|\nonumber\\
&+ \Big(F\beta^2+(1-F)(\gamma^2|b|^2+\delta^2|a|^2)-\frac{1}{2}\Big)
|1\rangle\langle1|\nonumber\\
&+(1-F)(\gamma^2-\delta^2)ab^\dag|0\rangle\langle
1|\nonumber\\
&+(1-F)(\gamma^2-\delta^2)a^\dag b|1\rangle\langle
0| =\M{0}.\label{eqn:munu_I3}
\end{align}
Therefore
\begin{align}
F\alpha^2+(1-F)(\gamma^2|a|^2+\delta^2|b|^2)-\frac{1}{2}&=0\label{eqn:munu_I3_1}\\
F\beta^2+(1-F)(\gamma^2|b|^2+\delta^2|a|^2)-\frac{1}{2}&=0\label{eqn:munu_I3_2}\\
(1-F)(\gamma^2-\delta^2)ab^\dag&=0\label{eqn:munu_I3_3}\\
(1-F)(\gamma^2-\delta^2)a^\dag b&=0.\label{eqn:munu_I3_4}
\end{align}

Since  $F<1$ and $\gamma<\delta$, from \eqref{eqn:munu_I3_3} and \eqref{eqn:munu_I3_4},
one can get $a=0$ or $b =0$. Without loss of generality, let $b=0$,
then $|a|=1$. Therefore \eqref{eqn:c01} becomes 
\begin{align}
|y_\mathrm{r}\rangle=e^{\imath \theta_{\mathrm{a}}}|0\rangle, \qquad |\tilde{y}_\mathrm{r}\rangle=e^{-\imath \theta_{\mathrm{a}}}|1\rangle\label{eqn:yr}
\end{align}
where $\theta_{\mathrm{a}}=\mathrm{pha}\{a\}.$
Substituting \eqref{eqn:yr} into \eqref{eqn:munu_I3_1} and \eqref{eqn:munu_I3_2} gives
\begin{align}
F\alpha^2 + (1-F)\gamma^2 =  F\beta^2 + (1-F)\delta^2 =\frac{1}{2}.\label{eqn:munu_I}
\end{align}
Since $F>\frac{1}{2}$, substituting \eqref{eqn:alpha_normal} into \eqref{eqn:munu_I} shows that
\begin{align}
\gamma<\beta<\frac{\sqrt 2}{2}<\alpha<\delta.\label{eqn:gammadeltainequal}
\end{align}

Moreover, substituting \eqref{eqn:yr} into \eqref{eqn:munuorth2} gives
\begin{align*}
e^{\imath \theta_{\mathrm{a}}}\alpha\gamma\langle0|z_\mathrm{r}\rangle
+e^{-\imath \theta_{\mathrm{a}}}\beta\delta\langle1|\tilde{z}_\mathrm{r}\rangle&=0
\end{align*}
which implies
\begin{align}
|\alpha\gamma\langle0|z_\mathrm{r}\rangle|
=|\beta\delta\langle1|\tilde{z}_\mathrm{r}\rangle|.\label{eqn:munuorth3}
\end{align}
On the other hand,
since  $\langle z_\mathrm{r}|\tilde{z}_\mathrm{r}\rangle=0$, it can be verified that $|\langle0|z_\mathrm{r}\rangle|=|\langle1|\tilde{z}_\mathrm{r}\rangle|$.
Therefore from \eqref{eqn:gammadeltainequal},
$|\alpha\gamma\langle0|z_\mathrm{r}\rangle|\le|\beta\delta\langle1|\tilde{z}_\mathrm{r}\rangle|$,
with the equality holds only if  $|\langle0|z_\mathrm{r}\rangle|=|\langle1|\tilde{z}_\mathrm{r}\rangle|=0$. This result together with  \eqref{eqn:munuorth3} implies $|z_\mathrm{r}\rangle=
e^{\imath \theta_{z}}|1\rangle$, $|\tilde{z_\mathrm{r}}\rangle=e^{\imath \theta_{\tilde{z}}}|0\rangle$, for some $\theta_{z}, \theta_{\tilde{z}}\in[0,2\pi)$.
Further noting that a quantum state is invariant up to an overall phase change, one can get
\begin{align}
|\nu\rangle = \gamma|01\rangle + \delta e^{\imath\theta} |10\rangle,\label{eqn:nu_01form}
\end{align}
where $\theta = \theta_{\tilde{z}}-\theta_{z}-2\theta_{\mathrm{a}}$. Equation \eqref{eqn:nu_01form} proves \eqref{eqn:nu}.
Therefore the local unitary operators $\M{U}_{\mathrm{A}}$, $\M{U}_{\mathrm{B}}$ exhibited in \eqref{eqn:UAB} give \eqref{eqn:int_dense_II}--\eqref{eqn:nu}.

Finally, we show that $\alpha$, $\beta$, $\gamma$, and $\delta$ satisfy \eqref{eqn:alphabeta_pzeta} and \eqref{eqn:gammadelta_pzeta}.
Substutite \eqref{eqn:mu2} and \eqref{eqn:nu_01form} into \eqref{eqn:checkrho}, then the density matrix of the second qubit $\check{\V{\rho}}_2=\mathrm{tr}_1\{\check{\V{\rho}}\}$ becomes
\begin{align}
\check{\V{\rho}}_2=\;&(F\alpha^2+(1-F)\delta^2)|0\rangle\langle0| + \nonumber\\
&(F\beta^2+(1-F)\gamma^2)|1\rangle\langle1|.
\label{eqn:checkrho2}
\end{align}

Noting that unitary operations does not change the determinant of a matrix, $\det\{\check{\V{\rho}}_2\}=\det\{\V{\rho}_2\}$.
Therefore, from \eqref{eqn:subden} and \eqref{eqn:checkrho2} one can get
\begin{align}
&\big(F\alpha^2 +(1-F)\delta^2\big) \big(F\beta^2 +(1-F)\gamma^2\big)\nonumber
\\&=\frac{1}{4}\big((1+|\eta|^2p)(1-p+\zeta^2p)-|\eta\zeta p|^2 \big). \label{eqn:alpha_pF2}
\end{align}
Substituting \eqref{eqn:alpha_normal} and \eqref{eqn:munu_I} into \eqref{eqn:alpha_pF2}, one can get
\begin{align*}
\alpha = \sqrt{\frac{1}{2}+\frac{|\eta| p}{4F}},\qquad \delta  = \sqrt{\frac{1}{2}+\frac{|\eta| p}{4(1-F)}}.
\end{align*}
This completes the proof.

\section{Proof of \thref{lem:simplifybound}}
\label{pf_lem:simplifybound}
Equation \eqref{eqn:F_lebound} holds trivially when $\gamma=0$ since fidelity of any qubit pair cannot exceed 1, i.e., $f(F, \alpha, \beta, \gamma, \delta, \theta) \le 1$. 
Hence, it remains to consider the case for which $0<\gamma\le\delta<1$, which will be proved by contradiction. 
Suppose \thref{lem:simplifybound} is false, then $\forall F\in(\frac{1}{2},1]$,
\begin{align}
f(F,\frac{1}{\sqrt{2}},\frac{1}{\sqrt{2}},\frac{1}{\sqrt{2}},\frac{1}{\sqrt{2}},0)\le\frac{F^2}{F^2+(1-F)^2}\label{eqn:rho_violateF}
\end{align}
but there exists some $F_0$, $\alpha_0$, $\beta_0$, $\gamma_0$, $\delta_0$ and $\theta_0$ such that
\begin{align}
f(F_0,\alpha_0,\beta_0,\gamma_0,\delta_0,\theta_0)>
\frac{F_0^2}{F_0^2+(1-F_0)^2(\frac{\gamma_0\delta_0}{\alpha_0\beta_0})^2}.\label{eqn:rho_violateF0}
\end{align}
Then contradiction would arise if there exist some $\tilde{F}\in(\frac{1}{2},1]$ such that \eqref{eqn:rho_violateF} does not hold.
To show the the existence of such $\tilde{F}$, the \ac{RSSP} method is employed to transform 
a given density matrix with parameters $F=\tilde{F}$, $\alpha=\beta=\gamma=\delta=\frac{1}{\sqrt{2}}$, and $\theta=0$ to another density matrix with parameters $F_0$, $\alpha_0$, $\beta_0$, $\gamma_0$, $\delta_0$ and $\theta_0$ via \ac{LOCC}.
In particular, consider that Alice measures her qubit using local operators
\begin{align}
\V{M}_{\mathrm A}=\begin{bmatrix}\sqrt{\frac{\alpha_0\gamma_0}{\beta_0\delta_0}}&0\\0&e^{\imath\frac{\theta_0}{2}}\end{bmatrix}, \quad \V{M}_{\bar{\mathrm A}}=\begin{bmatrix}\sqrt{1-\frac{\alpha_0\gamma_0}{\beta_0\delta_0}}&0\\0&
0\end{bmatrix}\label{eqn:MA}
\end{align}
and Bob measures his qubit using local operators
\begin{align}
\V{M}_{\mathrm B}=\begin{bmatrix}e^{\imath\frac{\theta_0}{2}}&0\\0&\sqrt{\frac{\beta_0\gamma_0}{\alpha_0\delta_0}}\end{bmatrix}, \quad \V{M}_{\bar{\mathrm B}}=\begin{bmatrix}0&0\\0&
\sqrt{1-\frac{\beta_0\gamma_0}{\alpha_0\delta_0}}\end{bmatrix}.\label{eqn:MB}
\end{align}
When the measurement results correspond to $\V{M}_{\mathrm A}$ and $\V{M}_{\mathrm B}$, the density matrix of the qubit pair after the measurement is given by
\begin{align}
\breve{\V{\rho}}&=\frac{\big(\V{M}_{\mathrm A}\otimes\V{M}_{\mathrm B}\big)\,\check{\V{\rho}}\,\big(\V{M}_{\mathrm A}\otimes\V{M}_{\mathrm B}\big)^\dag}{\mathrm{tr}\big\{\big(\V{M}_{\mathrm A}\otimes\V{M}_{\mathrm B}\big)\,\check{\V{\rho}}\,\big(\V{M}_{\mathrm A}\otimes\V{M}_{\mathrm B}\big)^\dag\big\}},\label{eqn:rho_pretransform}
\end{align}
where $\check{\V{\rho}}$ is the density matrix given in \eqref{eqn:int_dense_II}.
Set the channel to be phase-damping, i.e. $\eta=0$, then $\alpha=\beta=\gamma = \delta=\frac{1}{\sqrt{2}}$, and $\theta=0$. 
Further set the channel parameter $p$ so that $F$ equals to $\tilde{F}$ given by\footnote{Equation \eqref{eqn:alphabeta_pzeta}--\eqref{eqn:gammadelta_pzeta} imply $0\le \gamma\le \beta \le\frac{1}{\sqrt{2}}\le \alpha \le \delta\le 1$ and $\alpha^2 + \beta^2 = \gamma^2 + \delta^2 = 1$, showing that $\frac{\gamma\delta}{\alpha\beta}\in[0,1]$ for all valid $\alpha$, $\beta$, $\gamma$, and $\delta$. Hence, $\tilde{F}$ in \eqref{eqn:setp} is in the interval $(\frac{1}{2},1]$ as long as $F_0\in(\frac{1}{2},1]$. This guarantees the existence of $p$.}
\begin{align}\tilde{F}=\frac{F_0}{F_0+(1-F_0)\frac{\gamma_0\delta_0}{\alpha_0\beta_0}}.\label{eqn:setp}
\end{align}
Then according to \eqref{eqn:rho_pretransform}, the \ac{LOCC} for \ac{RSSP} transforms a density matrix $\check{\V{\rho}}$ with parameters $F=\tilde{F}$, $\alpha=\beta=\gamma=\delta=\frac{1}{\sqrt{2}}$, and $\theta=0$ to another density matrix given by 
\begin{align}
\breve{\V{\rho}}=\;&F_0 (\alpha_0|0\rangle+\beta_0|11\rangle)(\alpha_0\langle00|+\beta_0e\langle11|)\nonumber
\\&+(1-F_0)(\gamma_0|01\rangle+\delta_0e^{\imath \theta_0}|10\rangle)(\gamma_0\langle01|+\delta_0e^{-\imath \theta_0}\langle10|)\label{eqn:rho_transform}
\end{align}
whose parameters are $F_0$, $\alpha_0$, $\beta_0$, $\gamma_0$, $\delta_0$ and $\theta_0$. Let $\{\V{N}^{(k)}_{\mathrm{A}},\V{N}^{(k)}_{\mathrm{B}}\}_{k=1}^K$ be local operators that achieve the optimal fidelity $f(F_0,\alpha_0,\beta_0,\gamma_0,\delta_0,\theta_0)$ with initial density matrix $\breve{\V{\rho}}$.
Define new local operators
\begin{align*}
\V{L}^{(k)}_{\mathrm{A}}=\V{N}^{(k)}_{\mathrm{A}}\V{M}_{\mathrm A},\qquad 
\V{L}^{(k)}_{\mathrm{B}}=\V{N}^{(k)}_{\mathrm{B}}\V{M}_{\mathrm B}.
\end{align*}
Then $\{\V{L}^{(k)}_{\mathrm{A}},\V{L}^{(k)}_{\mathrm{B}}\}_{k=1}^K$ are valid local operators, and achieve the same fidelity $f(F_0,\alpha_0,\beta_0,\gamma_0,\delta_0,\theta_0)$ with initial density matrix $\check{\V{\rho}}$. Therefore, the optimal fidelity with initial density matrix $\check{\V{\rho}}$ is lower bounded by
\begin{align}
f(\tilde{F},\frac{1}{\sqrt{2}},\frac{1}{\sqrt{2}},\frac{1}{\sqrt{2}},\frac{1}{\sqrt{2}},0)&\ge f(F_0,\alpha_0,\beta_0,\gamma_0,\delta_0,\theta_0).\label{eqn:rho_violate_pre}
\end{align}
This together with \eqref{eqn:rho_violateF0} gives
\begin{align}
f(F,\frac{1}{\sqrt{2}},\frac{1}{\sqrt{2}},\frac{1}{\sqrt{2}},\frac{1}{\sqrt{2}},0)&>
\frac{F_0^2}{F_0^2+(1-F_0)^2(\frac{\gamma_0\delta_0}{\alpha_0\beta_0})^2}\nonumber
\\&=\frac{\tilde{F}^2}{\tilde{F}^2+(1-\tilde{F})^2}.\label{eqn:rho_violate}
\end{align}

With \eqref{eqn:rho_violate} the contradiction arises. This completes the proof.

\section{Proof of \thref{lem:operation}}
\label{pf_lem:operation}
Without loss of generality, denote the seperable operator acting on two qubit pairs as $\V{N}_{\mathrm{A}}\otimes\V{N}_{\mathrm{B}}$, where $\V{N}_{\mathrm{A}}$, $\V{N}_{\mathrm{B}}$ are employed by Alice and Bob respectively.
Every operator $\V{N}$ for two qubits  can be written equivalently as \mbox{$\V{N}=\V{N}(\V{H}^\dag\otimes\V{H}^\dag)(\V{H}\otimes\V{H}$)}, where  $\V{H}=\frac{1}{\sqrt{2}}\begin{bmatrix}1&1\\1&-1\end{bmatrix}$ is the Hadamard operator. Therefore, denote  $\tilde{\V{N}}_{\mathrm{X}}=\V{N}_{\mathrm{X}}(\V{H}^\dag\otimes\V{H}^\dag)$, $X\in\{A,B\}$, then the separable operator $\V{N}_{\mathrm{A}}\otimes\V{N}_{\mathrm{B}}$ for two qubit pairs is equivalent to first perform $\V{H}$ on every qubit, and then perform $\tilde{\V{N}}_{\mathrm{A}}\otimes\tilde{\V{N}}_{\mathrm{B}}$.

From \eqref{eqn:int_dense_II}, 
$\check{\V{\rho}}=F|\Phi^+\rangle\langle\Phi^+| + (1-F)|\Psi^+\rangle\langle\Psi^+|$ when the channel is phase-damping.
Hence after performing Hadamard operation on the qubits, the density matrix of the qubit pair becomes
\begin{align}
\tilde{\V{\rho}}&=(\V{H}\otimes\V{H})\,\check{\V{\rho}}\,(\V{H}\otimes\V{H})^\dag\nonumber\\
&=F|\Phi^+\rangle\langle\Phi^+| + (1-F)|\Phi^-\rangle\langle\Phi^-|.
\label{eqn:int_dense_VI}
\end{align}
Therefore, the joint density matrix of two qubit pair, where the first and last two qubits belong to Alice Bob respectively, is given by
\begin{align}
\V{\rho}_{\mathrm{J}}&=\V{P}\,(\tilde{\V{\rho}}\otimes\tilde{\V{\rho}})\,\V{P}\nonumber\\
%&=\V{P}\big(F|\Phi^+\rangle\langle\Phi^+| +(1-F)
%|\Phi^-\rangle\langle\Phi^-|\big)\otimes\big(F|\Phi^+\rangle\langle\Phi^+|
%+(1-F) |\Phi^-\rangle\langle\Phi^-|\big)\V{P}^\dag\nonumber\\
&=F^2|\Phi^{(1)}\rangle\langle\Phi^{(1)}| + F(1-F)(|\Phi^{(2)}\rangle\langle\Phi^{(2)}|
+ |\Phi^{(3)}\rangle\langle\Phi^{(3)}|)\nonumber
\\&\hspace{3.7mm}+ (1-F)^2|\Phi^{(4)}\rangle\langle\Phi^{(4)}|\label{eqn:initialjoint}
\end{align}
where $\V{P}$ is the permutation operator that switches the second and third qubits,
\begin{align}
%\label{eqn:Phi++}
\begin{split}
|\Phi^{(1)}\rangle&=\frac{1}{2}\big(|0000
\rangle+|0101\rangle+
|1010\rangle+
|1111\rangle\big)\\
%\label{eqn:Phi+-}
|\Phi^{(2)}\rangle&=\frac{1}{2}\big(|0000
\rangle-|0101\rangle+
|1010\rangle-
|1111\rangle\big)\\
%\label{eqn:Phi-+}
|\Phi^{(3)}\rangle&=\frac{1}{2}\big(|0000
\rangle+|0101\rangle-
|1010\rangle-
|1111\rangle\big)\\
%\label{eqn:Phi--}
|\Phi^{(4)}\rangle&=\frac{1}{2}\big(|0000
\rangle-|0101\rangle-
|1010\rangle+
|1111\rangle\big).
\end{split}\label{eqn:Phi1-4}
\end{align}
From \eqref{eqn:initialjoint},
after operator $\tilde{\V{N}}_{\mathrm{A}}\otimes\tilde{\V{N}}_B$ acts on the two qubit pairs,
the density matrix of the first qubit pair is given by
\begin{align}
\breve{\V{\rho}}&=\frac{\mathrm{tr}_{2,4}\big\{\sum_{i=1}^{4}C_i
\big(\tilde{\V{N}}_{\mathrm{A}}\otimes\tilde{\V{N}}_B\big)\,|\Phi^{(i)}\rangle\langle\Phi^{(i)}|\,
\big(\tilde{\V{N}}_{\mathrm{A}}\otimes\tilde{\V{N}}_B\big)^\dag\big\}}
{\mathrm{tr}\big\{\sum_{i=1}^{4}C_i
\big(\tilde{\V{N}}_{\mathrm{A}}\otimes\tilde{\V{N}}_B\big)\,|\Phi^{(i)}\rangle\langle\Phi^{(i)}|\,
\big(\tilde{\V{N}}_{\mathrm{A}}\otimes\tilde{\V{N}}_B\big)^\dag\big\}}\nonumber\\
&=\frac{\sum_{i=1}^{4}C_i\mathrm{tr}_{2,4}\big\{
\pmb{\phi}^{(i)}\pmb{\phi}^{(i)\dag}\big\}}
{\mathrm{tr}\big\{\sum_{i=1}^{4}C_i\mathrm{tr}_{2,4}\big\{\pmb{\phi}^{(i)\dag}
\pmb{\phi}^{(i)}\big\}\big\}}\label{eqn:rho2_s1}
\end{align}
where $C_1=F^2$, $C_2=C_3=F(1-F)$, $C_4=(1-F)^2$, and $\pmb{\phi}^{(i)}=\big(\tilde{\V{N}}_{\mathrm{A}}\otimes\tilde{\V{N}}_B\big)|\Phi^{(i)}\rangle$.
Denote $|w\rangle=|00\rangle$, $|x\rangle=|01\rangle$,
$|y\rangle=|10\rangle$, and $|z\rangle=|11\rangle$, and
denote
\begin{align*}
\pmb{\psi}^{(i)}=\mathrm{tr}_{2,4}\big\{
\pmb{\phi}^{(i)}\big\}.
\end{align*}
Then
% Then since separable state operated by , $\pmb{\phi}^{(i)}$ can be rewritten as
\begin{align}
\pmb{\psi}^{(i)}
&=\mathrm{tr}_{2,4}\big\{\big(\tilde{\V{N}}_{\mathrm{A}}\otimes\tilde{\V{N}}_B\big)\,|\Phi^{(i)}\rangle\}\nonumber\\
&=\big(\sum_{k=0}^1\big(\mathbb{I}_2\otimes\langle k|\big)\tilde{\V{N}}_{\mathrm{A}}\otimes\sum_{j=0}^1\big(\mathbb{I}_2\otimes\langle j|\big)\tilde{\V{N}}_B\big)\nonumber
\\&\hspace{5mm}\qquad\begin{bmatrix}|ww\rangle&
|xx\rangle&
|yy\rangle&
|zz\rangle
\end{bmatrix}\V{v}^{(i)}\nonumber\\
&=\begin{bmatrix}
\V{w}&
\V{x}&
\V{y}&
\V{z}
\end{bmatrix}\V{v}^{(i)}\label{eqn:Phi_2}
\end{align}
where $\V{v}^{(i)}$ is the $i$-th column of the unitary matrix $\V{V}$ defined in \eqref{eqn:U} and
\begin{align}
\V{s}&=\big(\sum_{k=0}^1\big(\mathbb{I}_2\otimes\langle k|\big)\tilde{\V{N}}_{\mathrm{A}}\otimes\sum_{j=0}^1\big(\mathbb{I}_2\otimes\langle j|\big)\tilde{\V{N}}_B\big)|ss\rangle\nonumber\\
&=\begin{bmatrix}
s_{11}\\
s_{12}\\
s_{21}\\
s_{22}
\end{bmatrix}\label{eqn:phi_2}
\end{align}
$s\in\{w,x,y,z\}$. %, and $\V{s}_\mathrm{A}=\tilde{\V{N}}_\mathrm{A}|s\rangle$, $\V{s}_\mathrm{B}=\tilde{\V{N}}_\mathrm{B}|s\rangle$, $s\in\{w,x,y,z\}$.
Combining  \eqref{eqn:rho2_s1}
and \eqref{eqn:Phi_2} gives \eqref{eqn:rho2} and \eqref{eqn:phi_3}. 

In \eqref{eqn:phi_2},
$|ss\rangle$ is a separable state, and 
\begin{align*}
\sum_{k=0}^1\big(\mathbb{I}_2\otimes\langle k|\big)\tilde{\V{N}}_{\mathrm{A}}\otimes\sum_{j=0}^1\big(\mathbb{I}_2\otimes\langle j|\big)\tilde{\V{N}}_B
\end{align*}
is a separable operator.
Therefore, vectors $\V{s}$, ${s}\in\{{w},{x},{y},{z}\}$ must also be separable.
As $1\times 4$ vectors, $\V{s}$ are separable if and only if
\begin{align*}
s_{11}s_{22}=s_{12}s_{21}, \qquad\forall s\in\{w,x,y,z\}
\end{align*} 
which give \eqref{eqn:srelation}. 

Finally, the probability that $\V{N}_{\mathrm{A}}\otimes\V{N}_{\mathrm{B}}$ acts on the qubits must not be $0$, which implies 
\begin{align*}\sum_{i=1}^{4}C_i\pmb{\psi}^{(i)\dag}\pmb{\psi}^{(i)}>0.
\end{align*} This completes the proof.

\section{Proof of \thref{thm:UBF}}
\label{pf_thm:UBF}
 %of the density matrix $\check{\V{\rho}}$ in \eqref{eqn:int_dense_II}.
We will first prove that the fidelity $F^*$ given in \eqref{eqn:Fidelity_Upper} is an upper bound.
From \thref{lem:simplifybound}, it is sufficient to prove the upper bound for the special case when the channel is phase-damping. 

Express the \ac{LOCC} performed by the agents as
$\V{N}^{(k)}_{\mathrm{A}}\otimes\V{N}^{(k)}_{\mathrm{B}}$, $k\in\{1,2,\ldots,K\}$.
Without loss of generality, assume $\V{N}^{(1)}_{\mathrm{A}}\otimes\V{N}^{(1)}_{\mathrm{B}}$ is one of the operators that lead to the highest fidelity.
Then from \thref{lem:operation}, conditioned on the event that 
$\V{N}^{(1)}_{\mathrm{A}}
\otimes\V{N}^{(1)}_{\mathrm{B}}$ acts on the two qubit pairs, the fidelity of the kept qubit pair is given by
\begin{align*}
&\langle\Phi^+|\breve{\V{\rho}}|\Phi^+\rangle\nonumber\\
&=
\frac{\sum_{i=1}^{4}C_i\langle\Phi^+|\,\pmb{\psi}^{(i)}\pmb{\psi}^{(i)\dag}|\Phi^+\rangle}
{\sum_{i=1}^{4}C_i\pmb{\psi}^{(i)\dag}\pmb{\psi}^{(i)}}\nonumber\\
&=\frac{\frac{1}{2}\sum_{i=1}^{4}C_i\Big|\sum_{k=1}^2\begin{bmatrix}w_{kk}
& x_{kk} & y_{kk} & z_{kk}\end{bmatrix}\V{v}^{(i)}
\Big|^2}
{\sum_{i=1}^{4}C_i\sum_{k=1}^2\sum_{j=1}^2\Big|\begin{bmatrix}w_{kj}
& x_{kj}& y_{kj} & z_{kj}\end{bmatrix}\V{v}^{(i)}
\Big|^2}
\end{align*}
where $C_i$, $\V{v}^{(i)}$, $i\in\{1,2,3,4\}$ and $s_{kj}$, $s\in\{w,x,y,z\}$, $k,j\in\{1,2\}$ are defined in \thref{lem:operation}.
Note that from \thref{thm:Str_rho}, $F\in[\frac{1}{2},1]$, implying that $C_1\ge C_2 = C_3 \ge C_4 \ge 0$.
Therefore, to prove the upper bound part of
\thref{thm:UBF}, it is sufficient to show that the following proposition
is true.

\begin{Prop}[Maximum fidelity]\thlabel{cor:FUB_1}
For any $s_{kj} \in\mathbb{C}$, $s\in\{w,x,y,z\}$, $k,j\in\{1,2\}$,
and $1\ge C_1\ge C_2\ge C_3 \ge C_4 \ge 0$, satisfying $s_{11}s_{22}=s_{12}s_{21}$ and 
\begin{align*}\sum_{i=1}^{4}C_i\sum_{k=1}^2\sum_{j=1}^2\Big|\begin{bmatrix}w_{kj}& x_{kj}& y_{kj} & z_{kj}\end{bmatrix}\V{v}^{(i)}\Big|^2>0\end{align*}
the following inequality holds
\begin{align}
& \frac{\sum_{i=1}^{4}C_i\Big|\sum_{k=1}^2
\begin{bmatrix}w_{kk}& x_{kk}
& y_{kk} & z_{kk}\end{bmatrix}\V{v}^{(i)} \Big|^2}
{\sum_{i=1}^{4}C_i\sum_{k=1}^2\sum_{j=1}^{2}\Big|\begin{bmatrix}w_{kj}&
x_{kj} & y_{kj} & z_{kj}\end{bmatrix}\V{v}^{(i)}\Big|^2 }\nonumber
\\&\le\frac{2C_1}{C_1+C_4}.
\label{eqn:Fidelity_3}
\end{align}
\end{Prop}

To prove the proposition above, first simplify  \eqref{eqn:Fidelity_3} via the following lemma.
\begin{Lem}[Simplify Parameters]\thlabel{lem:C} Consider coefficients $r_1, r_2,\check{r}_2,r_3,\check{r}_3,r_4\ge 0$, $\check{r}_4>0$ and variable $t\in[0,1]$, satisfying $r_3t + r_4>0$ and
\begin{align}
\qquad r_2\check{r}_4-\check{r}_2r_4\le0.\label{eqn:r}
\end{align}
If inequality
\begin{align}
\frac{r_1t + r_2}{r_3t + r_4}\le \frac{\check{r}_2}{\check{r}_3t + \check{r}_4}\label{eqn:FC_1}
\end{align}
holds for $t=\check{t}\ge0$, then it holds for all $t\in[0,\check{t}]$.
\end{Lem}
\begin{proof}
Define function 
\begin{align*}f(t)\triangleq r_1\check{r}_3 t^2 + (r_2\check{r}_3+ r_1\check{r}_4 - \check{r}_2r_3)t
+ r_2\check{r}_4 - \check{r}_2r_4.
\end{align*}
From \eqref{eqn:r}, $f(0)\le 0$.
Since $r_3t + r_4>0$ and $\check{r}_3t + \check{r}_4 > 0$, the fact that \eqref{eqn:FC_1} holds for $t=\check{t}$ is equivalent to  $f(\check{C})\le0$.
Moreover, since $f''(t)=r_1\check{r}_3\ge 0$, $f(t)$ is a convex function.
Therefore, $f(t)\le 0$, $\forall t\in[0,\check{t}]$, which is equivalent to \eqref{eqn:FC_1} holds $\forall t\in[0,\check{t}]$. This completes the proof
of \thref{lem:C}.
\end{proof}

Letting $C_4 = t$,
\begin{align*}
\Big|\sum_{k=1}^2\begin{bmatrix}w_{kk}& x_{kk}
& y_{kk} & z_{kk}\end{bmatrix}\V{v}^{(4)} \Big|^2  &= r_1\\
\sum_{i=1}^{3}C_i\Big|\sum_{k=1}^2\begin{bmatrix}w_{kk}& x_{kk}
& y_{kk} & z_{kk}\end{bmatrix}\V{v}^{(i)} \Big|^2 &= r_2\\
\sum_{k=1}^2\sum_{j=1}^{2}\Big|\begin{bmatrix}w_{kj}&
x_{kj} & y_{kj} & z_{kj}\end{bmatrix}\V{v}^{(4)}\Big|^2 &= r_3 \\
\sum_{i=1}^{3}C_i\sum_{k=1}^2\sum_{j=1}^{2}\Big|\begin{bmatrix}w_{kj}&
x_{kj} & y_{kj} & z_{kj}\end{bmatrix}\V{v}^{(i)}\Big|^2 &= r_4
\end{align*}
$2C_1= \check{r}_2$, $1=\check{r}_3$, and $C_1 = \check{r}_4$
in \eqref{eqn:Fidelity_3} gives the form of \eqref{eqn:FC_1}.
It can be verified that
\begin{align*}
&r_2\check{r}_4 - \check{r}_2r_4\nonumber\\
&=
C_1\sum_{i=1}^{3}C_i\bigg(\Big|\sum_{k=1}^2\begin{bmatrix}w_{kk}& x_{kk}
& y_{kk} & z_{kk}\end{bmatrix}\V{v}^{(i)} \Big|^2\nonumber\\
&\hspace{3.7mm}-2\sum_{k=1}^2\sum_{j=1}^2\Big|\begin{bmatrix}w_{kj}&
x_{kj} & y_{kj} & z_{kj}\end{bmatrix}\V{v}^{(i)}\Big|^2\bigg)\nonumber\\
&\le2C_1\sum_{i=1}^{3}C_i\bigg(\sum_{k=1}^2\Big|\begin{bmatrix}w_{kk}&
x_{kk} & y_{kk} & z_{kk}\end{bmatrix}\V{v}^{(i)} \Big|^2
\nonumber\\
&\hspace{3.7mm}-\sum_{k=1}^2\Big|\begin{bmatrix}w_{kk}&
x_{kk} & y_{kk} & z_{kk}\end{bmatrix}\V{v}^{(i)} \Big|^2\bigg)\nonumber\\
&=0. %\label{eqn:r_good1}
\end{align*}
Therefore, \thref{lem:C} shows that \eqref{eqn:FC_1} is true $\forall t\in[0,C_3]$ if it is true for $t=C_3$. This implies that to prove \thref{cor:FUB_1}, it is sufficient to prove \eqref{eqn:Fidelity_3} for the case of $C_4=C_3$.
Repeating this process two more times, i.e., appling \thref{lem:C} to \eqref{eqn:Fidelity_3} with $C_3=t$, and then with $C_2=t$, it can be shown that considering the case in which $C_1=C_2=C_3=C_4$ is sufficient to prove the proposition. Then, \eqref{eqn:Fidelity_3} simplifies to
\begin{align}
&\frac{\sum_{i=1}^{4}\left|\sum_{k=1}^2\begin{bmatrix}w_{kk}& x_{kk} & y_{kk}
& z_{kk}\end{bmatrix}\V{v}^{(i)} \right|^2}
{\sum_{k=1}^2\sum_{j=1}^{2}\sum_{i=1}^{4}\left|\begin{bmatrix}w_{kj}& x_{kj}
& y_{kj} & z_{kj}\end{bmatrix}\V{v}^{(i)}\right|^2 }\le1
\label{eqn:Fidelity_4}
\end{align}

Note that
\begin{align}
&\sum_{i=1}^{4}\left|\begin{bmatrix}w_{kj}& x_{kj} & y_{kj} & z_{kj}\end{bmatrix}\V{v}^{(i)}\right|^2\nonumber\\
&=\sum_{i=1}^{4}\begin{bmatrix}w_{kj}& x_{kj} & y_{kj} & z_{kj}\end{bmatrix}\V{v}^{(i)}\V{v}^{(i)\dag}\begin{bmatrix}w_{kj}& x_{kj} & y_{kj} & z_{kj}\end{bmatrix}^\dag \nonumber\\
&=\begin{bmatrix}w_{kj}& x_{kj} & y_{kj} & z_{kj}\end{bmatrix}\V{V}\V{V}^\dag\begin{bmatrix}w_{kj}& x_{kj} & y_{kj} & z_{kj}\end{bmatrix}^\dag \nonumber\\
&=\hspace{-2mm}\sum_{s\in\{w,x,y,z\}}\hspace{-2mm}|s_{kj}|^2\label{eqn:normrelation1}
\end{align}
where the last equality is due to the fact that $\V{V}$ is unitary.
Similarly,
\begin{align}
\sum_{i=1}^{4}\Big|\sum_{k=1}^2\begin{bmatrix}w_{kk}& x_{kk} & y_{kk} &
z_{kk}\end{bmatrix}\V{v}^{(i)} \Big|^2 = \hspace{-2mm}\sum_{s\in\{w,x,y,z\}}\hspace{-2mm}|s_{11}+s_{22}|^2.\label{eqn:normrelation2}
\end{align}
Then it can be obtained that
\begin{align}
&\sum_{k=1}^2\sum_{j=1}^{2}\sum_{i=1}^{4}\left|\begin{bmatrix}w_{kj}& x_{kj}
& y_{kj} & z_{kj}\end{bmatrix}\V{v}^{(i)}\right|^2\nonumber\\
&\stackrel{(a)}{=}  \hspace{-2mm}\sum_{s\in\{w,x,y,z\}}\hspace{-2mm}(|s_{11}|^2 + |s_{22}|^2) + \hspace{-2mm}\sum_{s\in\{w,x,y,z\}}\hspace{-2mm}(|s_{12}|^2
+ |s_{21}|^2)\nonumber\\
&\stackrel{(b)}{\ge}  \hspace{-2mm}\sum_{s\in\{w,x,y,z\}}\hspace{-2mm}(|s_{11}|^2 + |s_{22}|^2) + 2\hspace{-2mm}\sum_{s\in\{w,x,y,z\}}\hspace{-2mm}|s_{11}||s_{22}|
\nonumber\\
&\stackrel{(c)}{=}  \sum_{i=1}^{4}\Big|\sum_{k=1}^2\begin{bmatrix}w_{kk}& x_{kk} & y_{kk}
& z_{kk}\end{bmatrix}\V{v}^{(i)} \Big|^2\label{eqn:Fidelity_5}
\end{align}
where (a), (b), and (c) are due to \eqref{eqn:normrelation1}, \eqref{eqn:srelation}, and \eqref{eqn:normrelation2}, respectively.
This inequaltiy shows that \eqref{eqn:Fidelity_4} is true, which then proves \thref{cor:FUB_1}. This proofs that the fidelity $F^*$ given in \eqref{eqn:Fidelity_Upper} is an upper bound.

Finally, we use constructive method to show that fidelity $F^*$ given in \eqref{eqn:Fidelity_Upper} is achievable.
In fact, \eqref{eqn:F_update_FP1} of \thref{thm:dis_perf} shows that the fidelity in \eqref{eqn:Fidelity_Upper} is achieved by
adopting the \ac{RSSP} and first round distillation of the algorithm proposed in Section~\ref{subsec:achfid} and keeps a qubit pair only if measurement results correspond to $|1\rangle\langle1|$. This completes the proof.

\section{Proof of \thref{thm:dis_perf}}
\label{pf_thm:dis_perf}
First, the following lemma summarizes the effect of \ac{RSSP}.

\begin{Lem}[Performance of \ac{RSSP}]\thlabel{lem:single} In process of \ac{RSSP}, qubit pairs are kept with probability
\begin{align}
P_{\mathrm s}=2F_0\beta^2+(1-F_0)(\gamma^2+\frac{\beta^2\delta^2}{\alpha^2}).\label{eqn:prob_after1bit}
\end{align}
For a kept qubit pair, its density matrix is given by
\begin{align}
\tilde{\V{\rho}}= \tilde{F}|\Phi^+\rangle\langle\Phi^+| + (1-\tilde{F})|\tilde{\nu}\rangle\langle\tilde{\nu}|
\label{eqn:dense_after1bit}
\end{align}
where
\begin{align}
|\tilde{\nu}\rangle &= \tilde{\gamma}|01\rangle + \tilde{\delta}e^{\imath \theta} |10\rangle\label{eqn:nuhat}
\end{align}
with
\begin{align}
\tilde{F} &= \frac{2F_0\alpha^2\beta^2}{2F_0\alpha^2\beta^2+(1-F_0)(\alpha^2\gamma^2+\beta^2\delta^2)}\label{eqn:Fhat_pzeta}\\
\tilde{\gamma} &=\frac{\alpha\gamma}{\sqrt{\alpha^2\gamma^2+\beta^2\delta^2}}\label{eqn:gammahat_pzeta}\\
\tilde{\delta}  &=\frac{\beta\delta}{\sqrt{\alpha^2\gamma^2+\beta^2\delta^2}}.\label{eqn:deltahat_pzeta}
\end{align}
\end{Lem}
\begin{proof}
The qubit pairs are kept with probability
\begin{align}
\mathrm{tr}\{(\mathbb{I}_{2}\otimes\V{M}_{\mathrm{B}})  \,\check{\V{\rho}}\,  (\mathbb{I}_{2}\otimes\V{M}_{\mathrm{B}})^\dag\}\label{eqn:prob_1bitopt}
\end{align}
and the density matrix of a kept qubit pair is given by
\begin{align}
\tilde{\V{\rho}}&=\frac{(\mathbb{I}_{2}\otimes\V{M}_{\mathrm{B}})  \,\check{\V{\rho}}\,  (\mathbb{I}_{2}\otimes\V{M}_{\mathrm{B}})^\dag}
{\mathrm{tr}\big\{(\mathbb{I}_{2}\otimes\V{M}_{\mathrm{B}})  \,\check{\V{\rho}}\,  (\mathbb{I}_{2}\otimes\V{M}_{\mathrm{B}})^\dag\big\}}.
\label{eqn:dense_1bitopt}
\end{align}
Substituting \eqref{eqn:int_dense_II}--\eqref{eqn:gammadelta_pzeta} into \eqref{eqn:prob_1bitopt} and \eqref{eqn:dense_1bitopt},
one can obtain \eqref{eqn:dense_after1bit}--\eqref{eqn:deltahat_pzeta}. The details are omitted for brevity.
\end{proof}

From \eqref{eqn:dense_after1bit} and \eqref{eqn:nuhat}, after the \ac{RSSP}, the joint density matrix of two qubit pairs, where the first and last two qubits belong to Alice and Bob respectively, is given by
\begin{align*}
{\V{\rho}}_{\mathrm{J}}&=\V{P}\tilde{\V{\rho}}\otimes\tilde{\V{\rho}}\,\V{P}^\dag\\
&=\tilde{F}^2|\Omega^{(1)}\rangle\langle\Omega^{(1)}| +
\tilde{F}(1-\tilde{F})\big(|\Omega^{(2)}\rangle\langle\Omega^{(2)}|
+ |\Omega^{(3)}\rangle\langle\Omega^{(3)}|\big)\\
&\hspace{3.7mm}+ (1-\tilde{F})^2|\Omega^{(4)}\rangle\langle\Omega^{(4)}|
\end{align*}
where $\V{P}$ is the permutation operator that switches the second and third qubits, and
\begin{align}
\begin{array}{*{8}{r@{\;}}r@{}r}
|\Omega^{(1)}\rangle&=&&
\frac{1}{2}|0 0 0 0 \rangle&+&
\frac{1}{2}|0 1 0 1 \rangle
\vspace{1mm}\\
&&+&
\frac{1}{2}|1 0 1 0 \rangle&+&
\frac{1}{2}|1 1 1 1 \rangle
\vspace{1mm}\\
|\Omega^{(2)}\rangle&=&&
\frac{\tilde{\gamma}\sqrt{2}}{2}|0 0 0 1 \rangle&+&
\frac{\tilde{\delta}e^{\imath \theta}\sqrt{2}}{2}|0 1 0 0 \rangle
\vspace{1mm}\\
&&+&
\frac{\tilde{\gamma}\sqrt{2}}{2}|1 0 1 1 \rangle&+&
\frac{\tilde{\delta}e^{\imath \theta}\sqrt{2}}{2}|1 1 1 0 \rangle
\vspace{1mm}\\
|\Omega^{(3)}\rangle&=&&
\frac{\tilde{\gamma}\sqrt{2}}{2}|0 0 1 0 \rangle&+&
\frac{\tilde{\gamma}\sqrt{2}}{2}|0 1 1 1 \rangle
\vspace{1mm}\\
&&+&
\frac{\tilde{\delta}e^{\imath \theta}\sqrt{2}}{2}|1 0 0 0 \rangle&+&
\frac{\tilde{\delta}e^{\imath \theta}\sqrt{2}}{2}|1 1 0 1 \rangle
\vspace{1mm}\\
|\Omega^{(4)}\rangle&=&&
\tilde{\gamma}^2|0 0 1 1 \rangle&+&
\tilde{\gamma}\tilde{\delta}e^{\imath \theta}|0 1 1 0 \rangle
\vspace{1mm}\\
&&+&
\tilde{\gamma}\tilde{\delta}e^{\imath \theta}|1 0 0 1 \rangle&+&
\tilde{\delta}^2e^{\imath 2\theta}|1 1 0 0 \rangle&&.
\end{array}
\nonumber
\end{align}

For the first round of distillation, after both agents perform the CNOT operation, the joint density matrix of two qubit pairs becomes
\begin{align}
\check{\V{\rho}}_{\mathrm{J}}
&=\tilde{F}^2|\check\Omega^{(1)}\rangle\langle\check\Omega^{(1)}| +
\tilde{F}(1-\tilde{F})\big(|\check\Omega^{(2)}\rangle\langle\check\Omega^{(2)}|
+ |\check\Omega^{(3)}\rangle\langle\check\Omega^{(3)}|\big)\nonumber
\\&\hspace{3.7mm}
+ (1-\tilde{F})^2|\check\Omega^{(4)}\rangle\langle\check\Omega^{(4)}|\label{eqn:jointBCNOT}
\end{align}
where\begin{align}
\begin{array}{*{8}{r@{\;}}r@{}r}
|\check\Omega^{(1)}\rangle&=&&
\frac{1}{2}|0 0 0 0 \rangle&+&
\frac{1}{2}|0 1 0 1 \rangle
\vspace{1mm}\\
&&+&
\frac{1}{2}|1 1 1 1 \rangle&+&
\frac{1}{2}|1 0 1 0 \rangle\vspace{1mm}
\\
|\check\Omega^{(2)}\rangle&=&&
\frac{\tilde{\gamma}\sqrt{2}}{2}|0 0 0 1 \rangle&+&
\frac{\tilde{\delta}e^{\imath \theta}\sqrt{2}}{2}|0 1 0 0 \rangle
\vspace{1mm}\\
&&+&
\frac{\tilde{\gamma}\sqrt{2}}{2}|1 1 1 0 \rangle&+&
\frac{\tilde{\delta}e^{\imath \theta}\sqrt{2}}{2}|1 0 1 1 \rangle
\vspace{1mm}\\
|\check\Omega^{(3)}\rangle&&=&
\frac{\tilde{\gamma}\sqrt{2}}{2}|0 0 1 1 \rangle&+&
\frac{\tilde{\gamma}\sqrt{2}}{2}|0 1 1 0 \rangle
\vspace{1mm}\\
&&+&
\frac{\tilde{\delta}e^{\imath \theta}\sqrt{2}}{2}|1 1 0 0 \rangle&+&
\frac{\tilde{\delta}e^{\imath \theta}\sqrt{2}}{2}|1 0 0 1 \rangle\vspace{1mm}\\
|\check\Omega^{(4)}\rangle&&=&
\tilde{\gamma}^2|0 0 1 0 \rangle&+&
\tilde{\gamma}\tilde{\delta}e^{\imath \theta}|0 1 1 1 \rangle
\vspace{1mm}\\
&&+&
\tilde{\gamma}\tilde{\delta}e^{\imath \theta}|1 1 0 1 \rangle&+&
\tilde{\delta}^2e^{\imath 2\theta}|1 0 0 0 \rangle&.
\end{array}
\nonumber
\end{align}

From \eqref{eqn:jointBCNOT}, if both measurement results correspond to $|1\rangle\langle1|$, the (unnormalized) density matrix of the source qubit pair is given by
\begin{align}
\V{\rho}_{11}&=(\mathbb{I}_2\otimes\langle1|\otimes\mathbb{I}_2\otimes\langle1
|)\,\check{\V{\rho}}_{\mathrm{J}}\,(\mathbb{I}_2\otimes|1\rangle\otimes\mathbb{I}_2\otimes|1\rangle)\nonumber
\\&=\tilde{F}^2\frac{1}{2}|\Phi^+\rangle\langle\Phi^+|\nonumber
\\&\hspace{3.7mm}+ (1-\tilde{F})^2(\tilde{\gamma}\tilde{\delta})^2\big(|01\rangle+|10\rangle\big)\big(\langle01|+\langle10|\big).
\label{eqn:case11}
\end{align}
Otherwise, if both measurement results correspond to $|0\rangle\langle0|$, the (unnormalized) density matrix of the source qubit pair is given by
\begin{align}
\V{\rho}_{00}&=(\mathbb{I}_2\otimes\langle0|\otimes\mathbb{I}_2\otimes\langle0
|)\,\check{\V{\rho}}_{\mathrm{J}}\,(\mathbb{I}_2\otimes|0\rangle\otimes\mathbb{I}_2\otimes|0\rangle)\nonumber
\\\nonumber &=\tilde{F}^2\frac{1}{2}|\Phi^+\rangle\langle\Phi^+|+ (1-\tilde{F})^2\big(\tilde{\gamma}^2|01\rangle+\tilde{\delta}^2e^{\imath 2\theta}|10\rangle\big)
\\&\hspace{3.7mm}\big(\tilde{\gamma}^2\langle01|+\tilde{\delta}^2e^{-\imath 2\theta}\langle10|\big)
\label{eqn:case00}
\end{align}

From \eqref{eqn:case11}, and \eqref{eqn:case00}, if the agents adopt the \ac{FP} approach, i.e., keep the source qubit pair only if both measurement results correspond to $|1\rangle\langle1|$, the probability of keeping the source qubit pair is
\begin{align}
P_{\mathrm{f}}=\mathrm{tr}\{\V{\rho}_{11}\}=\frac{\tilde{F}^2}{2}+2(1-\tilde{F})^2(\tilde{\gamma}\tilde{\delta})^2\label{eqn:Pf}
\end{align}
the fidelity of the kept qubit pairs is
\begin{align}
F_1=\frac{\frac{1}{2}\tilde{F}^2}{P_f}=\frac{\tilde{F}^2}{\tilde{F}^2+4(1-\tilde{F})^2(\tilde{\gamma}\tilde{\delta})^2}\label{eqn:F1f}
\end{align}
and the density matrix of the kept qubit pair is
\begin{align}
\tilde{\V{\rho}}=\frac{\V{\rho}_{11}}{P_f}=F_1|\Phi^+\rangle\langle\Phi^+| +(1-F_1)|\Psi^+\rangle\langle\Psi^+|.\label{eqn:rho1f}
\end{align}

If the agents adopt the \ac{PP} approach, i.e.,  keeping the source qubit pair if the measurement results match, the probability of preserving the source qubit pair is
\begin{align}
P_{\mathrm{p}}=\mathrm{tr}\{\V{\rho}_{11}+\V{\rho}_{00}\}=\tilde{F}^2+(1-\tilde{F})^2\label{eqn:Pp}
\end{align}
the fidelity of the kept qubit pairs is
\begin{align}
F_1=\frac{\frac{1}{2}\tilde{F}^2+\frac{1}{2}\tilde{F}^2}{P_1}=\frac{\tilde{F}^2}{\tilde{F}^2+(1-\tilde{F})^2}\label{eqn:F1p}
\end{align}
and the density matrix of the kept qubit pair can be written as
\begin{align}
\tilde{\V{\rho}}&=\frac{\V{\rho}_{11}+\V{\rho}_{00}}{P_{\mathrm{p}}}\nonumber
\\&=F_1|\Phi^+\rangle\langle\Phi^+| +G|\Psi\rangle\langle\Psi| + \tilde{G}|\tilde{\Psi}\rangle\langle\tilde{\Psi}|\label{eqn:rho1p}
\end{align}
where $G+\tilde{G}=1-F_1$, $|\Psi\rangle, |\tilde{\Psi}\rangle\in\mathrm{span}(|01\rangle, |10\rangle)$, and $\langle\Psi|\tilde{\Psi}\rangle=0$.

From \thref{lem:single}, \eqref{eqn:Pf}, \eqref{eqn:F1f}, \eqref{eqn:Pp} and \eqref{eqn:F1p}, after the \ac{RSSP} and first round of distillation, a qubit pair is kept with probability
\begin{align*}
P_1=\left\{
\begin{array}{llr}
\dfrac{P_{\mathrm s}P_{\mathrm f}}{2} &=&\dfrac{F_0^2\alpha^2\beta^4 + (1-F_0)^2\beta^2\gamma^2\delta^2}{2F_0\alpha^2\beta^2 + (1-F_0)(\alpha^2\gamma^2 +\beta^2\delta^2)}\vspace{1mm}\\
&&\mbox{for the \ac{FP} approach}\vspace{2mm}\\
\dfrac{P_{\mathrm s}P_{\mathrm p}}{2}&=&\dfrac{4F_0^2\alpha^4\beta^4 + (1-F_0)^2(\alpha^2\gamma^2 +\beta^2\delta^2)^2}{4F_0\alpha^4\beta^2 + 2(1-F_0)\alpha^2(\alpha^2\gamma^2 +\beta^2\delta^2)}\vspace{1mm}\\
&&
\mbox{for the \ac{PP} approach}
\end{array}
\right.
\end{align*}
and fidelity
\begin{align*}
F_1=\left\{\begin{array}{l@{\;}l}
\dfrac{F_0^2}{F_0^2+(1-F_0)^2(\frac{\gamma\delta}{\alpha\beta})^2}
&
\mbox{for the \ac{FP} approach}\vspace{1.5mm}\\
\dfrac{F_0^2}{F_0^2+\frac{1}{4}(1-F_0)^2(\frac{\gamma^2}{\beta^2}+\frac{\delta^2}{\alpha^2})^2}
&
\mbox{for the \ac{PP} approach}
\end{array}\right.
\end{align*}

For the following rounds of distillations, one can take \eqref{eqn:rho1f} or \eqref{eqn:rho1p} as input, use similar analysis as in \eqref{eqn:jointBCNOT}--\eqref{eqn:case00} and \eqref{eqn:Pp}--\eqref{eqn:rho1p}. This analysis will show that
\begin{align*}
P_{k}&=\frac{1}{2}\big(F_{k-1}^2+(1-F_{k-1})^2\big),\\
F_k&=\frac{F_{k-1}^2}{F_{k-1}^2+(1-F_{k-1})^2}\end{align*}
and the density matrix of the kept qubit pairs maintains the same structure as in  \eqref{eqn:rho1f} or \eqref{eqn:rho1p}.
This competes the proof.

\end{document}